\documentclass[11pt, letterpaper]{article}
\usepackage[utf8]{inputenc}
\usepackage{fullpage,times}

\usepackage{url}

\usepackage{amsmath,amsfonts,amsthm,amssymb,multirow}
\usepackage{times}
\usepackage{graphicx}
\usepackage{floatpag}
\usepackage{algorithm}
\usepackage[noend]{algpseudocode}
\usepackage{bbold}
\usepackage{enumitem}
\usepackage{subcaption} 
\usepackage{calc}
\usepackage{tikz}
\usetikzlibrary{decorations.markings,decorations.pathreplacing,arrows, automata, positioning, quotes, shapes,backgrounds, arrows.meta, shapes.multipart}
\tikzstyle{vertex}=[circle, draw, inner sep=0pt, minimum size=4pt, fill = black]
\usepackage[numbers,sort]{natbib}

\usepackage{graphicx}
\usepackage{tabularx}
\makeatletter
\newcommand{\multiline}[1]{%
  \begin{tabularx}{\dimexpr\linewidth-\ALG@thistlm}[t]{@{}X@{}}
    #1
  \end{tabularx}
}
\makeatother
\usepackage{mathtools}

\makeatletter
\def\BState{\State\hskip-\ALG@thistlm}
\makeatother

\usepackage[compact]{titlesec}
\titlespacing{\section}{0pt}{3ex}{2ex}
\titlespacing{\subsection}{0pt}{2ex}{1ex}
\titlespacing{\subsubsection}{0pt}{0.5ex}{0ex}

\newtheorem{theorem}{Theorem}[section]
\newtheorem{fact}[theorem]{Fact}

\newtheorem{proposition}[theorem]{Proposition}
\newtheorem{corollary}[theorem]{Corollary}

\newtheorem{definition}[theorem]{Definition}

\newtheorem{claim}[theorem]{Claim}

\newtheorem{conjecture}[theorem]{Conjecture}

\makeatletter
\let\c@fconjecture\c@conjecture
\makeatother

\makeatletter
\let\c@fconj\c@conj
\makeatother

\def \eps {\varepsilon}

\newcommand{\ignore}[1]{}

\def\tO{\tilde{O}}

\newcommand{\IGNORE}[1]{}

\title{Monochromatic Triangles, Triangle Listing and APSP}

\author{Virginia {Vassilevska Williams}\\MIT\\virgi@mit.edu \and Yinzhan Xu\\MIT\\xyzhan@mit.edu}
\date{}

\begin{document}

\maketitle

\begin{abstract}
All-Pairs Shortest Paths (APSP) is one of the most basic problems in graph algorithms. In one of the most general variants of the problem, one is given an $n$-node directed or undirected graph with integer weights in $\{-n^c,\dots,n^c\}$ and no negative cycles and one needs to compute the shortest paths distance between every pair of vertices. A central question in graph algorithms is how fast APSP can be solved. The fastest known algorithm runs in $n^3/2^{\Theta(\sqrt{\log n})}$ time [Williams'14], and no truly subcubic time algorithms are known.

One of the main hypotheses in fine-grained complexity is that this problem requires $n^{3-o(1)}$ time. Another famous hypothesis in fine-grained complexity is that the $3$SUM problem for $n$ integers (which can be solved in $O(n^2)$ time) requires $n^{2-o(1)}$ time. Although there are no direct reductions between $3$SUM and APSP, it is known that they are related: there is a problem, $(\min,+)$-convolution that reduces in a fine-grained way to both, and a problem Exact Triangle that both fine-grained reduce to.

In this paper we find more relationships between these two problems and other basic problems. P{\u{a}}tra{\c{s}}cu 
had shown that under the $3$SUM hypothesis the All-Edges Sparse Triangle problem in $m$-edge graphs requires $m^{4/3-o(1)}$ time. 
The latter problem asks to determine for every edge $e$, whether $e$ is in a triangle. It is equivalent to the problem of listing $m$ triangles in an $m$-edge graph where $m=\tO(n^{1.5})$, and can be solved in $O(m^{1.41})$ time [Alon et al.'97] with the current matrix multiplication bounds, and in $\tO(m^{4/3})$ time if $\omega=2$.

We show that one can reduce Exact Triangle to All-Edges Sparse Triangle, showing that All-Edges Sparse Triangle (and hence Triangle Listing) requires $m^{4/3-o(1)}$ time also assuming the APSP hypothesis. This allows us to provide APSP-hardness for many dynamic problems that were previously known to be hard under the $3$SUM hypothesis.

We also consider the previously studied All-Edges Monochromatic Triangle problem. Via work of [Lincoln et al.'20], our result on All-Edges Sparse Triangle implies that if the All-Edges Monochromatic Triangle problem has an $O(n^{2.5-\eps})$ time algorithm for $\eps>0$, then both the APSP and $3$SUM hypotheses are false. The fastest algorithm for All-Edges Monochromatic Triangle runs in $\tO(n^{(3+\omega)/2})$ time [Vassilevska et al.'06], and our new reduction shows that if $\omega=2$, this algorithm is best possible, unless $3$SUM or APSP can be solved faster.
Besides $3$SUM, previously, the only problems known to be fine-grained reducible to
All-Edges Monochromatic Triangle were the seemingly easier problems directed unweighted APSP and Min-Witness Product [Lincoln et al.'20]. Our new reduction shows that this problem is much harder. We also connect the problem to other ``intermediate'' problems, whose runtimes are between $O(n^\omega)$ and $O(n^3)$, such as the Max-Min product problem.
\end{abstract}
\newpage

% The  
% had shown that $3$SUM can be reduced in truly subquadratic time to the All-Edges Sparse Triangle problem of determining for every edge $e$ in an $n$-node, $m=\tO(n^{1.5})$-edge graph, whether $e$ is in a triangle. 
% a triangle that $e$ is in, if it is in one. \todo{All-Edges Sparse Triangle in Lincoln et. al'20 is the detection version of this problem. This listing version was called Triangle Reporting I think?}

% \section{Preliminaries}
% \input{prelim}

\section{Introduction}

All-Pairs Shortest Paths (APSP) is one of the most fundamental problems in graph algorithms. In one of the most general variants of the problem, one is given an $n$-node directed or undirected graph with integer weights in $\{-n^c,\dots,n^c\}$ with no negative cycles and one needs to compute the shortest paths distance between every pair of vertices. A central question in graph algorithms is how fast APSP in $n$ node graphs can be solved. The fastest known algorithm runs in $n^3/2^{\Theta(\sqrt{\log n})}$ time~\cite{Williams14a,Williams18}, and no  $O(n^{3-\eps})$ time for $\eps>0$, so called {\em truly subcubic} algorithms are known.

One of the main hypotheses in fine-grained complexity, the APSP hypothesis, is that APSP requires $n^{3-o(1)}$ time\footnote{All of the hypotheses are for the Word-RAM model of computation with $O(\log n)$ bit words.}. Another famous hypothesis is the $3$SUM hypothesis that the $3$SUM problem for $n$ integers in  $\{-n^c,\ldots,n^c\}$ requires $n^{2-o(1)}$ time. $3$SUM asks if three of the integers sum to $0$ and can be solved in $O(n^2)$ time. Although there are no direct reductions between $3$SUM and APSP, it is known that they are related,  in a sense much more related than the third core problem of fine-grained complexity Orthogonal Vectors (OV). For instance, the $(\min,+)$-convolution problem is known to be fine-grained reducible to both $3$SUM and APSP, so that if APSP is in truly subcubic time, or $3$SUM is in truly subquadratic time, then $(\min,+)$-convolution is in truly subquadratic time \cite{CyganMWW19,patrascu2010towards,VWfindingcountingj,VassilevskaW09,focsy,focsyj}. Further, the Exact Triangle problem is known to be harder than both APSP and $3$SUM \cite{VWfindingcountingj,focsyj}, so that if it has a truly subcubic time algorithm, then both the APSP and the $3$SUM hypotheses would be false.
Meanwhile, it is not known how $3$SUM and APSP (or $(\min,+)$-convolution and Exact Triangle) are related to OV\footnote{The Min-Weight $k$-Clique problem can be reduced to both APSP and \textit{moderate-dimension} OV \cite{abboud2018more}, but it is not known whether it can be reduced to OV.}.

In this paper we provide more relationships to other problems that $3$SUM and APSP have in common. P{\u{a}}tra{\c{s}}cu \cite{patrascu2010towards} showed that $3$SUM can be reduced in truly subquadratic time to the All-Edges Sparse Triangle problem of determining for every edge $e$ in an $n$-node, $m=\tO(n^{1.5})$-edge\footnote{The $\tO$ notation in this paper hides poly-logarithmic factors.} graph, whether $e$ is in a triangle.
P{\u{a}}tra{\c{s}}cu's reduction implies that under the $3$SUM hypothesis, All-Edges Sparse Triangle requires $m^{4/3-o(1)}$ time.
All-Edges Sparse Triangle is known to be equivalent to the problem of listing $m$ triangles in an $m$-edge graph where $m=\tO(n^{1.5})$ \cite{duraj2020equivalences}, and can be solved in $O(m^{1.41})$ time \cite{AlonYZ97} with the current matrix multiplication bounds \cite{vstoc12,legallmult}, and in $\tO(m^{4/3})$ time if the exponent of square matrix multiplication $\omega$ is $2$.

Our main result is a reduction from Exact Triangle to All-Edges Sparse Triangle, thus showing that All-Edges Sparse Triangle (and hence Triangle Listing) requires $m^{4/3-o(1)}$ time also assuming the APSP hypothesis. This allows us to provide APSP-hardness for many dynamic problems.

%\todo{start using AE-Mono$\Delta$}
We also consider the previously studied All-Edges Monochromatic Triangle problem (AE-Mono$\Delta$) \cite{VassilevskaWY10,VassilevskaWY06,lincoln2020monochromatic} in which one is given an $n$-node graph $G$ with colors on the edges, and one is asked to return for every edge $e$, whether it appears in a monochromatic triangle in $G$. The fastest algorithm for AE-Mono$\Delta$ runs in $\tO(n^{(3+\omega)/2})$ time \cite{VassilevskaWY10,VassilevskaWY06}.
Via work of \cite{lincoln2020monochromatic}, our reduction from Exact Triangle to All-Edges Sparse Triangle implies that if AE-Mono$\Delta$ has an $O(n^{2.5-\eps})$ time algorithm for $\eps>0$, then both the APSP and $3$SUM hypotheses are false. This shows that if $\omega=2$, the known algorithm for AE-Mono$\Delta$ is best possible, unless $3$SUM and APSP can both be solved faster. 

\cite{lincoln2020monochromatic} showed that $3$SUM is in fact fine-grained equivalent to the {\em Monochromatic Convolution} problem, which is the convolution version of AE-Mono$\Delta$. These two latter problems are very related (e.g. the known algorithms for them are analogous, the only difference being the use of FFT vs Fast Matrix Multiplication), to the extent that one might conjecture that they are fine-grained equivalent. If this bold conjecture were true, then $3$SUM would be equivalent to both problems, and thus APSP would reduce to $3$SUM. We leave determining the veracity of this conjecture to future work.

Besides $3$SUM, previously, the only problems known to be fine-grained reducible to
AE-Mono$\Delta$ were the seemingly easier problems directed unweighted APSP and Min-Witness Product \cite{lincoln2020monochromatic}. Our new reduction shows that this problem is much harder. 
 We also connect the problem to multiple so-called ``intermediate'' problems, whose runtimes are between $O(n^\omega)$ and $O(n^3)$, such as the Max-Min Product problem that has been studied in relation to All Pairs Bottleneck Paths \cite{VassilevskaWY07,duanpettiebott}.
 
\subsection{Our results}
Here we give more details about the results summarized above.

% =============== Reductions From Exact Triangle ========

\paragraph{Reductions from Exact Triangle.}
The Exact Triangle problem asks, given an $n$-node graph with integer edge weights in $\{-n^c,\ldots,n^c\}$ for some constant $c$ and a target $T$, whether there are three vertices $p,q,r$ so that $w(p,q)+w(q,r)+w(r,p)=T$. Exact Triangle is equivalent to the version of Exact Triangle in which the target $T$ is $0$ \cite{VWfindingcountingj}. This is called the Zero Triangle problem
(also known as Zero-$3$-Clique or Love Triangle).

The brute-force algorithm for Exact Triangle runs in $O(n^3)$ time.
Meanwhile, as mentioned earlier, an $O(n^{3-\eps})$ time algorithm for $\eps>0$ for Exact Triangle would violate both the APSP hypothesis and the $3$SUM hypothesis. The {\em Exact Triangle hypothesis} states that Exact Triangle requires $n^{3-o(1)}$ time in the word-RAM model of computation with $O(\log n)$ bit words. This hypothesis is at least as believable as the hypothesis that at least one of the $3$SUM and APSP hypotheses is true.

Our main result is a reduction from Exact Triangle to certain unbalanced versions of Triangle Listing, and All-Edges Triangle Listing in sparse graphs which can then easily be reduced to other problems, including All-Edges Sparse Triangle, AE-Mono$\Delta$, SetDisjointness and SetIntersection.

In our unbalanced triangle problems there are five parameters, $\alpha,\beta,\gamma,\rho,t$. One is given a tripartite graph where the three parts have  $n^\alpha, n^\beta, n^\gamma$ vertices, respectively. The fourth parameter $\rho$ controls the edge density of the graph. Roughly speaking, each vertex in the graph has an $n^{-\rho}$ fraction of vertices as its neighbors. In the parameterized version of All-Edges Triangle Listing one is asked to list for every edge $e$, $t$ triangles containing $e$ (or all triangles containing $e$ if there are fewer than $t$ triangles). In our parameterized version of Triangle Listing, one is asked to list $t$ triangles in the graph (or all triangles if there are fewer than $t$ triangle). 
%In the parametreized version of All-Edges Triangle Detection the parameter $t$ is omitted, and one needs to return for every edge in the graph whether it is contained in a triangle.

The statements of our reductions to the parameterized versions of Triangle Listing and All-Edges Triangle Listing are a bit technical (see Theorem~\ref{thm:exact-to-general-triangle-listing} in the Section~\ref{sec:exact_triangle}). We will list some consequences.

% \todo{Can we avoid defining this messy problem in the intro? }
% We consider an unbalanced version of triangle listing in tripartite graphs, where the number of vertices in the three parts are $n^\alpha, n^\beta, n^\gamma$ respectively. There is also a parameter $\rho$ that controls the density of edges in the graph. Roughly speaking, each vertex in the graph has $n^{-\rho}$ portion of vertices as its neighbors. 

% Our main reduction reduces Exact Triangle to this unbalanced version of triangle listing with various settings of parameters. This main reduction implies hardness  for a variety of problems, including All-Edges Sparse Triangle, All-Edges Monochromatic Triangles, SetDisjointness and SetIntersection. \todo{Here as is written is correct, but we could also mention that we have a lemma that reduces listing $O(1)$ triangles for each  edge to detecting whether each edge is in a triangle, in a graph of the same structure.}

\begin{corollary}
\label{cor:intro:exact-to-sparse}
Any $n$-node instance of Exact Triangle can be reduced in $\tilde{O}(n^{2.5})$ time to $\tilde{O}(n)$ instances of All-Edges Sparse Triangle on $O(n^{1.5})$ edges each. Thus, assuming the Exact Triangle hypothesis, there is no $O(m^{4/3-\epsilon})$ time algorithm for All-Edges Sparse Triangle for any $\epsilon > 0$. 
\end{corollary}

Lincoln et al. \cite{lincoln2020monochromatic} showed that computing a certain number of  independent instances of All-Edges Sparse Triangle is equivalent to AE-Mono$\Delta$. 
\begin{theorem}\label{thm:sparse_to_mono}\cite{lincoln2020monochromatic}
Computing $m^{1/3}$ independent instances of All-Edges Sparse Triangle with $m$ edges each, where the number of vertices in each instance is $m^{2/3}$ is equivalent up to poly-logarithmic factors to computing an AE-Mono$\Delta$ instance where the number of vertices is $n=O(m^{2/3})$.
\end{theorem}

Combining Corollary~\ref{cor:intro:exact-to-sparse} and Theorem~\ref{thm:sparse_to_mono}, we get the following conditional lower bound for AE-Mono$\Delta$. 

\begin{corollary}
\label{cor:intro:mono-triangle-hard}
Assuming the Exact Triangle hypothesis, there is no $O(n^{2.5 -\epsilon})$ algorithm for AE-Mono$\Delta$ on $n$-node graphs for any $\epsilon>0$.
\end{corollary}

We also achieve the following conditional lower bound for even sparser triangle detection problems, as another corollary of our main reduction.

\begin{corollary}
Let $\mathbb{A}$ be an algorithm for All-Edges Sparse Triangle for  $n$-node graphs where every node has degree at most $d=n^\delta$ for some $0 < \delta \le 0.5$.
Assuming the Exact Triangle hypothesis, $\mathbb{A}$ cannot run in $O(n^{1-\epsilon}d^2)$ time for any $\epsilon > 0$. 
\end{corollary}

\paragraph{Many $3$SUM-hard problems are now also APSP-hard.} 
Kopelowitz, Pettie and Porat \cite{kopelowitz2016higher} considered the SetDisjointness and SetIntersection problems. They showed $3$SUM hardness for both SetDisjointness and  SetIntersection, which were in turn used to show $3$SUM hardness for many dynamic graph problems. 

The SetDisjointness problem can be viewed as an All-Edges Sparse Triangle problem in constrained graphs, and the SetIntersection problem can be viewed as a triangle listing problem. Using our reductions from Exact Triangle to All-Edges Sparse Triangle and Triangle Listing, we show that the SetDisjointness and SetIntersection problems are hard under the Exact Triangle hypothesis. Thus, we immediately get APSP-hardness for a variety of problems that were previously known to be $3$SUM-hard,  including Dynamic Maximum Cardinality Matching, Incremental Maximum Cardinality Matching, $d$-Failure Connectivity and Triangle Enumeration in graphs with particular arboricity \cite{kopelowitz2016higher}. One example result is the following. 

% \begin{corollary}
% Assume the $3$SUM Conjecture or the APSP Conjecture. There is no algorithm for Maximum Cardinality Matching that runs in expected $O(n^{\frac{\sqrt{17}-1}{8} - \epsilon})$ for $\epsilon > 0$ where $n$ is the size of the graph. 
% \end{corollary}

\begin{corollary}
Assume the Exact Triangle hypothesis (or any one of the $3$SUM hypothesis or the APSP hypothesis). For any constants $y \in (0, 1/2)$, $x \in (0, 2y]$, there exists a constant $\epsilon > 0$ so that for graphs with $n$ vertices, $m$ edges, arboricity $\alpha = \Theta(n^x) = \Theta(m^y)$, and $t < m \alpha^{1-\epsilon}$ triangles, it requires $\Omega(m\alpha^{1-o(1)})$ time to list all triangles in the graph. 
\end{corollary}

Notably, the SetIntersection problem of Kepolowitz et al. is a more generalized version of the All-Edges Sparse Triangle problem considered by P{\u{a}}tra{\c{s}}cu \cite{patrascu2010towards} and later by Abboud and Williams \cite{abboud2014popular}. Therefore, all conditional lower bound results in \cite{patrascu2010towards} and \cite{abboud2014popular} that use All-Edges Sparse Triangle as a source of hardness also hold under the Exact Triangle hypothesis. These results
include (but are not limited to) dynamic $st$ reachability, dynamic $st$ shortest paths, dynamic strong connectivity, subgraph connectivity, Langerman’s problem, Pagh’s problem and Erickson’s problem.
An example result from All-Edges Sparse Triangle is the following. 

\begin{corollary}
Assume the Exact Triangle hypothesis (or any one of the $3$SUM hypothesis or the APSP hypothesis). Then there is no fully dynamic algorithm for $s$-$t$ reachability that can have $O(m^{4/3-\epsilon})$ processing time, $O(m^{a-\epsilon})$ update time, and $O(m^{2/3-a-\epsilon})$ query time for some $\epsilon > 0$ and $1/6 \le a \le 1/3$. 
\end{corollary}

% Henzinger et al.~\cite{HenzingerKNS15} introduced the Online Matrix Vector Multiplication (OMv) hypothesis and proved many conditional hardness results from it. Practically all of the hardness results in the paper go through an intermediate problem OuMv which is an online version of All-Nodes Triangle Detection.

Readers interested in reductions from triangle problems to dynamic graphs problems can check \cite{kopelowitz2016higher, patrascu2010towards, abboud2014popular} for more details. 

% \todo{Check if Henzinger et al's stuff from OMV also gets lower bounds from APSP; these will be weaker but maybe still something}

% =============== Starting Second Chunk of the Reductions ========

\paragraph{Reductions for intermediate problems.}
The notion of  ``intermediate'' problems were first devised by Lincoln et al. \cite{lincoln2020monochromatic} who studied problems whose current best running time is $\tilde{O}(n^{2.5})$ if $\omega = 2$. With the current bounds on rectangular matrix multiplication \cite{legallurr}, the running time of these problems vary. The easiest ``intermediate'' problems seem to be the unweighted directed APSP \cite{zwickbridge} problem and the Min-Witness Product problem \cite{CzumajKL07}, whose best algorithms run in $\tilde{O}(n^{2.5286})$ time using the best bounds on rectangular matrix multiplication \cite{legallurr}. Following them in complexity, there are the Equality Product problem \cite{vnotes,labib2019hamming} and the Dominance Product problem \cite{MatIPL,YusterDom}, which are known to have $\tilde{O}(n^{2.6598})$ time algorithms. 

The Equality Product of two $n\times n$ integer matrices $A$ and $B$ is the $n\times n$ matrix $C$ with $C[i,j]=|\{k~|~A[i,k]=B[k,j]\}|$. The Boolean version of the problem, $\exists$Equality Product asks to determine whether $|\{k~|~A[i,k]=B[k,j]\}|>0$, for each $i,j$. The Dominance Product of integer matrices $A$ and $B$ is the matrix $C$ with $C[i,j]=|\{k~|~A[i,k]\leq B[k,j]\}|$. The Boolean version of the problem, $\exists$Dominance Product asks whether $|\{k~|~A[i,k]\leq B[k,j]\}|>0$, for each $i,j$. While the regular versions of Equality Product and Dominance Product are known to be equivalent \cite{vnotes,labib2019hamming}, their Boolean versions are not known to be, although $\exists$Dominance Product can be reduced to $O(\log n)$ instances of $\exists$Equality Product \cite{vnotes,labib2019hamming}.

There are also several intermediate problems that do not have any improvement brought by fast rectangular matrix multiplication, and whose running times are all $\tilde{O}(n^{(3+\omega)/2}) = \tilde{O}(n^{2.6865})$. There problems include for instance,  the aforementioned AE-Mono$\Delta$ problem, the $(\min, \le)$-product problem, the Max-Min Product problem studied in relation to the All-Pairs Bottleneck Paths and All-Pairs Nondecreasing Paths \cite{duanpettiebott,VassilevskaWY07,DuanJW19}, and the related $(\min, =)$-product which we introduce in this paper.

% \todo{also mention the counting version of equality and dominance are equivalent? }

Lincoln et al. \cite{lincoln2020monochromatic} gave reductions from unweighted directed APSP to Max-Min Product and AE-Mono$\Delta$, showing that if there exists a $T(n)$ time algorithm for Max-Min Product or AE-Mono$\Delta$, then one can also obtain an $\tilde{O}(T(n))$ time algorithm for unweighted directed APSP. They also give reductions from Min-Witness Product to Max-Min Product and AE-Mono$\Delta$. 

These reductions are not tight when $\omega > 2$, since they are reductions from a seemingly easier problem (for which improvements via rectangular matrix multiplication are known) to a harder problem (for which no improvements via rectangular matrix multiplication are known). 
For instance, in order to use the reduction to AE-Mono$\Delta$ to get a better algorithm for unweighted directed APSP, one would
need to obtain a better than $\tilde{O}(n^{2.5286})$ time algorithm for AE-Mono$\Delta$. This doesn't seem doable with current techniques.
%If we regard the reduction as a lower bound, we obtain that if there is no $O(n^{2.5-\epsilon})$ time algorithm for $\epsilon >0$ for unweighted directed APSP, then there is also no $O(n^{2.5-\epsilon})$ time algorithm for $\epsilon >0$ for AE-Mono$\Delta$. 
% However, 
% since unweighted directed APSP is an easy ``intermediate'' problem, the assumption is less likely to be true. 
Hence a natural question is whether one can obtain tight reductions between some of the intermediate problems.

% target to get tight reductions between some ``intermedaite'' problems. 

As a first result in this direction, we show that AE-Mono$\Delta$ is the hardest ``intermediate'' problem among the ones mentioned above, as all of the problems can be reduced to it.
% I removed "tightly" since only those whose current running time does not improve by rectangular MM can be reduced to it tightly.. 
A key step in the reduction is to study a new problem called \textit{All-Edges Monochromatic Equality Triangle} (AE-MonoEq$\Delta$), which can be viewed as a combination of AE-Mono$\Delta$ and $\exists$Equality Product. We delay its formal definition to Section~\ref{sec:prelim}. As the first step in the reductions, we reduce AE-Mono$\Delta$ to AE-MonoEq$\Delta$. 

\begin{theorem}
\label{thm:intro:mono->monoeq}
If AE-Mono$\Delta$ has an $\tilde{O}(n^{(3+\omega)/2-\epsilon})$ time algorithm for $0 \le \epsilon \le \frac{\omega}{2} - 1$, then AE-MonoEq$\Delta$ has an $\tilde{O}(n^{(3+\omega)/2-\delta})$ time algorithm for $\delta \ge 0$. Moreover, if $\epsilon > 0$ then $\delta > 0$. 
\end{theorem}

If we use $\epsilon = 0$ in 
Theorem~\ref{thm:intro:mono->monoeq}, 
we immediately get an $\tilde{O}(n^{(3+\omega)/2})$ time algorithm for AE-MonoEq$\Delta$, showing that AE-MonoEq$\Delta$ is indeed another ``intermediate'' problem. If $\omega > 2$, we can use $\epsilon > 0$ in Theorem~\ref{thm:intro:mono->monoeq} to get that any slight improvements over the current best algorithm for AE-Mono$\Delta$ implies an improved algorithm for  AE-MonoEq$\Delta$. 

As suggested by the names of the problems, AE-Mono$\Delta$ is a special version of AE-MonoEq$\Delta$ (See Section~\ref{sec:intermediate} for the formal definition), so that any runtime improvements for AE-MonoEq$\Delta$ imply runtime improvements of AE-Mono$\Delta$. Thus, AE-MonoEq$\Delta$ and AE-Mono$\Delta$ are sub-$n^{(3+\omega)/2}$-fine-grained equivalent when $\omega > 2$. 

Next, we show that one can reduce all the ``intermediate'' problems mentioned above
to AE-MonoEq$\Delta$. 

\begin{theorem}
\label{thm:intro:monoeq->everything}
Suppose there exists a $T(n)$ time algorithm for AE-MonoEq$\Delta$, then there exist $\tilde{O}(T(n))$ time algorithms for all of the following:
% \todo{better format?}
\begin{itemize}\setlength\itemsep{0em}
    \item unweighted directed APSP,
    \item Min-Witness Product,
    \item $\exists$Equality Product,
    \item $\exists$Dominance Product,
    \item $(\min, =)$-product,
    \item Max-Min Product,
    \item $(\min, \le)$-product.
\end{itemize}
\end{theorem}

Some reductions in the above theorem were already known. First of all, both Max-Min Product and $\exists$Dominance Product can be reduced to
$(\min, \le)$-product \cite{VassilevskaWY07}. Next, both Min Witness Product and unweighted directed APSP can be reduced to Max-Min Product \cite{lincoln2020monochromatic}. As mentioned earlier, $\exists$Dominance Product reduces to $\exists$Equality Product, and the latter easily reduces to $(\min, =)$-product. Thus, to
prove Theorem~\ref{thm:intro:mono->monoeq}, it suffices to give reductions from $(\min, =)$-product and $(\min, \le)$-product to AE-MonoEq$\Delta$. 

Since AE-MonoEq$\Delta$ has an $\tilde{O}(n^{(3+\omega) / 2})$ time algorithm, Theorem~\ref{thm:intro:mono->monoeq} provides alternative algorithms for $(\min, =)$-product and Max-Min Product. These new algorithms are also potentially simpler as they do not involve dealing with sparse matrix products, which were the main source of difficulty in the previous $\tilde{O}(n^{(3+\omega) / 2})$ time algorithms for the problems
%\todo{Simplifies previous algorithm, maybe? At least in some way.. i.e., if we use the AE-MonoEq$\Delta$ framework, we don't have to deal with sparse matrices, which is an annoying thing in the previous algorithm for $(\min, \max)$}

Combining Theorem~\ref{thm:intro:mono->monoeq} and Theorem~\ref{thm:intro:monoeq->everything}, we obtain that AE-Mono$\Delta$ is the hardest ``intermediate'' problem among all ``intermediate'' problems considered in \cite{lincoln2020monochromatic}, in the sense that if there is any improvement of AE-Mono$\Delta$ over the $\tilde{O}(n^{(3+\omega)/2})$ running time when $\omega > 2$, there will also be improvements for  $(\min, =)$-product, Max-Min Product and $(\min, \le)$-product.

AE-MonoEq$\Delta$ can be viewed as many independent instances of a problem called AE-Eq$\Delta$, in which we are given a graph $G$ with edge values, and we are asked to decide for each edge $e$ in the graph, whether it is in a triangle such that at least two of its three edges share the same edge value. 
Via techniques used in the proof of Theorem~$4$ in \cite{lincoln2020monochromatic}, we can show that computing a single instance of AE-MonoEq$\Delta$ of size $n$ is equivalent to, up to poly-logarithmic factors,  computing a certain number of instances of AE-Eq$\Delta$ on graphs with $n$ vertices where the total number of edges across all instances is $\Theta(n^2)$. 

% The main idea of Theorem~$4$ in \cite{lincoln2020monochromatic} is that we can 1) randomly permute the vertex sets of each instance independently; 2) embed them into the same graph and argue that we can avoid collisions via some techniques; 3) use edge colors to denote which instance a particular edge comes from. Thus, AE-MonoEq$\Delta$ itself is also a natural problem. 

Motivated by the simple nature of AE-Eq$\Delta$ and its relationship to our AE-MonoEq$\Delta$ problem, we consider the monochromatic versions of other intermediate problems.
The most interesting of these are arguably the monochromatic versions of $\exists$Equality Product and $(\min, =)$-product which we call Monochromatic Equality Product and Monochromatic $(\min, =)$-product. 

In the Monochromatic Equality Product problem (MonoEq), we are given a tripartite graph $G$ on vertex parts $I \cup J \cup K$. Each edge $e$ in $G$ has a color $c(e)$. All edges $e$ in $I \times K$ and $J \times K$ has a value $v(e) $. For every $(i, j)$, we need to decide if there exists $k$ such that $v[i, k] = v[j, k]$ and $c[i, k] = c[j, k] = c[i, j]$. MonoEq can be regarded as a combination of many sparse $\exists$Equality Product instances, where we are given sparse matrices $A$ and $B$, and we need to compute their $\exists$Equality Product result on a small number of entries. MonoEq is a special case of AE-MonoEq$\Delta$, so there exists an $\tilde{O}(n^{(3+\omega)/2})$ time algorithm for it.

The input to the Monochromatic $(\min, =)$-product (MonoMinEq) problem is the same as the input to MonoEq. For the output, instead of only determining for each $(i,j)$ the existence of $k$ such that $v[i, k] = v[j, k]$ and $c[i, k] = c[j, k] = c[i, j]$, we also have to output the minimum value of such a $v[i, k]$. MonoMinEq can be viewed as combination of many sparse $(\min, =)$-product instances. 

The best known algorithms for $\exists$Equality Product and $(\min, =)$-product have different running times, and it is unclear whether they are equivalent; clearly $\exists$Equality reduces to $(\min, =)$-product, but a reduction in the other direction would imply an improvement over the known algorithms for $(\min, =)$-product.

Surprisingly, we are able to show that the monochromatic versions,  MonoEq and MonoMinEq, are equivalent up to poly-logarithmic factors. 

\begin{theorem}
\label{thm:intro:monoeqproduct}
If there exists a $T(n)$ time algorithm for Monochromatic Equality Product, then there exists an $\tilde{O}(T(n))$ time algorithm for Monochromatic $(\min, =)$-product, and vice versa.
\end{theorem}

Theorem~\ref{thm:intro:monoeqproduct} also implies an $\tilde{O}(n^{(3+\omega)/2})$ time algorithm for MonoMinEq. 

\section{Preliminaries}
\label{sec:prelim}
In this section, we recall formal definitions of problems considered in this paper and define notations used in the proofs. 

\subsection{Hardness Sources}
\begin{definition}[$3$SUM]
Given $n$ integers in  $\{-n^c,\ldots,n^c\}$ for constant $c$, determine  if three of the integers sum to $0$.
\end{definition}
\begin{conjecture}[$3$SUM hypothesis]
In the word-RAM model with $O(\log n)$ bit words, there is no $O(n^{2-\epsilon})$ for $\epsilon > 0$ time algorithm for $3$SUM.  
\end{conjecture}

\begin{definition}[APSP]
Given an $n$-node directed graph with integer weights in $\{-n^c,\dots,n^c\}$ and no negative cycles, compute the shortest paths distance between every pair of vertices.
\end{definition}
\begin{conjecture}[APSP hypothesis]
In the word-RAM model with $O(\log n)$ bit words, there is no $O(n^{3-\epsilon})$ for $\epsilon > 0$ time algorithm for APSP.  
\end{conjecture}

\begin{definition}[Exact Triangle, Exact$\Delta$]
Given an $n$-node graph with integer edge weights in $\{-n^c,\ldots,n^c\}$ for some constant $c$ and a target $T$, decide whether there are three vertices $p,q,r$ so that $w(p,q)+w(q,r)+w(r,p)=T$.
\end{definition}

\begin{conjecture}[Exact Triangle hypothesis]
In the word-RAM model with $O(\log n)$ bit words, there is no $O(n^{3-\epsilon})$ for $\epsilon > 0$ time algorithm for Exact Triangle.  
\end{conjecture}

It is known that either the $3$SUM hypothesis or the APSP hypothesis implies the Exact Triangle hypothesis \cite{VWfindingcountingj,focsyj}.

\begin{definition}[Zero Triangle, Zero$\Delta$]
Given an $n$-node graph with integer edge weights in $\{-n^c,\ldots,n^c\}$ for some constant $c$, decide whether there are three vertices $p,q,r$ so that $w(p,q)+w(q,r)+w(r,p)=0$
\end{definition}

It is known that Exact Triangle is sub-cubic fine-grained equivalent to Zero Triangle \cite{VWfindingcountingj}.

\subsection{Graph Problems}

\begin{definition}[All-Edges Sparse Triangle, AE-Sparse$\Delta$]

Given an $n$-node $m$-edge graph $G=(V,E)$, 
determining for every edge $e \in E$   whether $e$ is in a triangle. 
\end{definition}

\begin{definition}[All-Pairs Shortest Paths in directed unweighted graphs, UnweightedAPSP]
Given an $n$-node directed unweighted graph $G=(V,E)$, 
compute the shortest paths distance between every pair of vertices.
\end{definition}

\subsection{Set Problems}

we define two problems investigated in \cite{kopelowitz2016higher}: SetDisjointness and SetIntersection. 

\begin{definition}[SetDisjointness]
Given a universe $U$, a family $\mathcal{F} \subseteq 2^U$ of subsets of $U$, and $q$ pairs of queries $(S, S') \in \mathcal{F} \times \mathcal{F}$, determine for each query $(S, S')$ whether $S \cap S'$ is empty.
\end{definition}

\begin{definition}[SetIntersection]
Given a universe $U$, a family $\mathcal{F} \subseteq 2^U$ of subsets of $U$, $q$ pairs of queries $(S, S') \in \mathcal{F} \times \mathcal{F}$ and a number $T$, output elements of $S \cap S'$ for each query $(S, S')$. It is allowed to terminate the algorithm once it outputs $T$ elements in total.  
\end{definition}

\subsection{Matrix Product Problems}

\begin{definition}[Minimum Witness matrix product, MinWitness]
Given two $n \times n$ Boolean matrices $A, B$, compute an $n \times n$ matrix $C$ such that 
$$C_{i, j} = \min( \{ k: A_{i, k} = B_{k, j} = 1\} \cup \{\infty\}).$$
\end{definition}
\begin{definition}[Max-Min Product, $(\max, \min)$-product]
Given two $n \times n$ integer matrices $A, B$, compute an $n \times n$ matrix $C$ such that 
$$C_{i, j} = \max_k \min \left\{ A_{i, k}, B_{k, j}\right\} .$$
\end{definition}

We define a generic $(\oplus, \otimes)$-product, where $\oplus$ maps a set of integers to an integer, and $\otimes$ maps two integers to a Boolean value. 

\begin{definition}[$(\oplus, \otimes)$-product]
Given two $n \times n$ integer matrices $A, B$, compute an $n \times n$ matrix $C$ such that 
$$C_{i, j} = \oplus( \{ B_{k, j} : A_{i, k} \otimes B_{k, j}\}_k ).$$
\end{definition}

% Equality Product can be defined as $(|\cdot|, =)$-product, where $|\cdot|$ returns the size of a set of integers. Similarly, Dominance Product can be defined as $(|\cdot|, \le)$-product.

We can define operation \textbf{NonEmpty} that returns $1$ if a set is nonempty, and $0$ otherwise. We can therefore define $\exists$Equality Product as $(\textbf{NonEmpty}, =)$-product, and define $\exists$Dominance Product as $(\textbf{NonEmpty}, \le)$-product.

If we define $\min(\emptyset) = \infty$ and $\max(\emptyset) = -\infty$, then 
$(\min, =)$-product, $(\min, \le)$-product  and $(\max, \le)$-product all fall into this generic definition without ambiguity. 

\subsection{Problems with Colors}

\begin{definition}[All-Edges Monochromatic  Triangle, AE-Mono$\Delta$]
Given an $n$-node  graph $G = (V, E)$, where each edge $e \in E$ has a color $c(e)$. Determine  for every edge $e$, whether it appears in a monochromatic triangle in $G$.
\end{definition}

\begin{definition}[All-Edges Monochromatic Equality Triangle, AE-MonoEq$\Delta$]

Given an $n$-node graph $G = (V, E)$, where each edge $e \in E$ has a color $c(e)$ and a value $v(e)$. Determine  for every edge $e$, whether it appears in a monochromatic triangle in $G$ that at least two of its edges share the same value.  
\end{definition}

\begin{definition}[Monochromatic Equality Product, MonoEq]
Given a graph $G = (I \cup J \cup K, E)$, where $|I|=|J| = |K| = n$. Each edge $e$ in the graph has a color $c(e)$. Each edge $e$ in $I \times K$ and $J \times K$ has a value $v(e) $. For every $(i, j)$, decide if there exists $k$ such that $v(i, j) = v(j, k)$ and $c(i, k) = c(j, k) = c(i, j)$.
\end{definition}

\begin{definition}[Monochromatic $(\min, =)$-product, MonoMinEq]
Given a graph $G = (I \cup J \cup K, E)$, where $|I|=|J| = |K| = n$. Each edge $e$ in the graph has a color $c(e)$. Each edge $e$ in $I \times K$ and $J \times K$ has a value $v(e) $. For every $(i, j)$, compute 
$$\min \left( \{v(i, k) : v(i, k) = v(j, k) \wedge c(i, k) = c(j, k) = c(i, j)\}_k \cup \{\infty\} \right).$$
\end{definition}

\begin{definition}[Monochromatic $(\min, \le)$-product]
Given a graph $G = (I \cup J \cup K, E)$, where $|I|=|J| = |K| = n$. Each edge $e$ in the graph has a color $c(e)$. Each edge $e$ in $I \times K$ and $J \times K$ has a value $v(e) $. For every $(i, j)$, compute 
$$\min  \left( \{v(j, k) : v(i, k) \le v(j, k) \wedge c(i, k) = c(j, k) = c(i, j)\}_k \cup \{\infty\} \right).$$
\end{definition}

\subsection{Notations}

For a graph $G=(V,E)$, a node $v\in V$ and $U\subseteq V$, we use $\deg(v, U)$ to denote $|\{u \in U : (v, u) \in E\}|$. 

We use $\omega(a, b, c)$ to denote the rectangular matrix multiplication exponent, i.e. the smallest real number $z$ such that the time to multiply an $n^a \times n^b$ matrix by an $n^b \times n^c$ matrix is $O(n^{z+\eps})$ for all $\eps>0$.
In particular, let $\omega=\omega(1,1,1)$ be the exponent for square matrix multiplication. It is known that $\omega\in [2,2.373)$ \cite{vstoc12,legallmult}. The best known bounds for $\omega(a,b,c)$ are in \cite{legallurr}.

\section{Hardness of Sparse Triangle Listing}
\label{sec:exact_triangle}
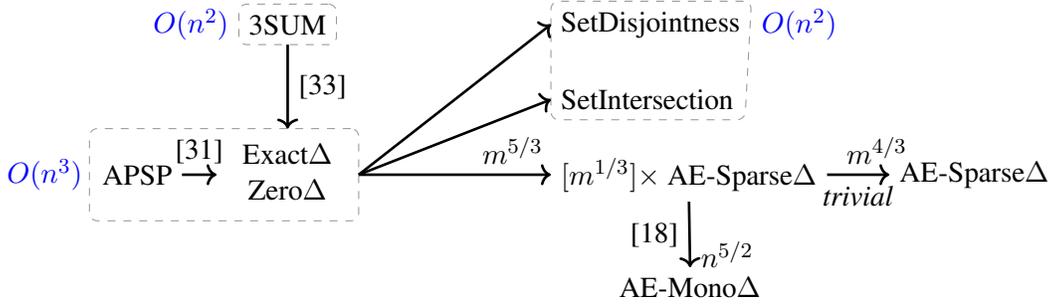
\begin{figure}[ht]
    \centering
    \begin{tikzpicture}
    
        \pgfmathsetmacro{\dx}{-7.35}
        \pgfmathsetmacro{\dy}{1.5}
        
        \node at(2 + \dx, 2 + \dy)  [anchor=center] (3sum){3SUM};
        \node at(3sum.west)  [anchor=east] (){\textcolor{blue}{$O(n^2)$}};
        
        \node at(0 + \dx, 0 + \dy)  [anchor=center] (apsp){APSP};
         \node at(apsp.west)  [anchor=east] (){\textcolor{blue}{$O(n^3)$}};
         
         \node at(0 + \dx, 0 + \dy)  [anchor=center] (apspholder){\begin{tabular}{c}  \\ \ \ \ \ \ \ \ \end{tabular}};
         
        \node at(2 + \dx, 0 + \dy)  [anchor=center] (exact){\begin{tabular}{c} Exact$\Delta$ \\ Zero$\Delta$ \end{tabular}};
        
        \node at(5.5 + \dx, 0 + \dy)  [anchor=west] (sparses){$[m^{1/3}] \times $ AE-Sparse$\Delta$};
        
        \node at(10 + \dx, 0 + \dy)  [anchor=west] (sparse){AE-Sparse$\Delta$};
        
        \node at(7.35 + \dx, -1.5 + \dy)  [] (monoT){AE-Mono$\Delta$};

        \node at(5.5 + \dx, 1 + \dy)  [anchor=west] (inter){SetIntersection};
        
        \node at(5.5 + \dx, 2 + \dy)  [anchor=west] (disj){SetDisjointness};
        \node at(disj.east)  [anchor=west] (){\textcolor{blue}{$O(n^2)$}};

        %  \node at(0, -4)  [anchor=center] (exact){Exact$\Delta$};
         
        %  \node at(4.5, 0)  [anchor=west] (AElisting){AE $\Delta$ Listing};
        %  \node at(4.5, 2)  [anchor=west] (listing){$\Delta$Listing};
         
        %  \node at(7.25, 0.5)  [anchor=west] (AEdetection){AE-$\Delta$Detection};
         
        %  \node at(10.5, 0)  [anchor=west] (mono){AE-Mono$\Delta$};
        %  \node at(10.5, 1)  [anchor=west] (disjoint){SetDisjointness};
        %   \node at(10.5, 2)  [anchor=west] (intersect){SetIntersection};

        \draw[->,line width=1pt] (3sum.south) to[]  node[right] {\cite{VWfindingcountingj}} (exact.north);
        \draw[->,line width=1pt] (apsp.east) to[]  node[above] {\cite{focsyj}} (exact.west);
        \draw[->,line width=1pt] (exact.east) to[]  node[] {} (disj.west);
        \draw[->,line width=1pt] (exact.east) to[]  node[] {} (inter.west);
        \draw[->,line width=1pt] (exact.east) to[]  node[] {} node[pos=0.8,above]{$m^{5/3}$}(sparses.west);
        
        \draw[->,line width=1pt] (sparses.east) to[]  node[below] {\textit{trivial}} node[pos=0.8,above]{$m^{4/3}$}(sparse.west);
        
        \draw[->,line width=1pt] (sparses.south) to[]  node[left] {\cite{lincoln2020monochromatic}} node[pos=0.8,right]{$n^{5/2}$}(monoT.north);
        
        \draw[opacity=0.4, dashed, rounded corners=3] (3sum.north east) -- (3sum.north west) -- (3sum.south west) -- (3sum.south east) -- cycle;

        \draw[opacity=0.4, dashed, rounded corners=3] (exact.north east) -- (apspholder.north west) -- (apspholder.south west) -- (exact.south east) -- cycle;
        
         \draw[opacity=0.4, dashed, rounded corners=3] (disj.north east) -- (disj.north west) -- (inter.south west) -- (inter.south east) -- cycle;

    \end{tikzpicture}
    \caption{The main reduction in Section~\ref{sec:exact_triangle} is from Exact$\Delta$/Zero$\Delta$ to parameterized versions of triangle detection and triangle listing. This reduction implies hardness for SetDisjointness, SetIntersection, AE-Sparse$\Delta$ and AE-Mono$\Delta$.}
    \label{fig:exact_triangle_reductions}
\end{figure}

We define a parameterized version of Triangle Listing. In this version, the graph has three parts of vertices. Each of the three parts can have different sizes, but edge densities between any pair of parts are the same. 

\begin{definition}[$(\alpha, \beta, \gamma, \rho, t)$-All-Edges Triangle Listing]
\label{def:AE-triangle-listing}
Given a tripartite graph $G$ whose vertex set is $A \cup B \cup C$, such that $|A| = n^\alpha, |B| = n^\beta, |C| = n^\gamma$. Let $X, Y \in \{A, B, C\}$ be two different parts of the graph. For any $v \in X$, $\deg(v, Y) \le O(n^{-\rho} |Y|)$. The $(\alpha, \beta, \gamma, \rho, t)$-All-Edges Triangle Listing problem asks to list, for each $e \in E \cap (A \times B)$, all triangles containing $e$ if there are fewer than $t$ such triangles or $t$ distinct triangles containing $e$ is there are at least $t$ such triangles.
\end{definition}

Definition~\ref{def:AE-triangle-listing} defines the triangle listing problem slightly differently from the usual definition. In previous works (e.g. \cite{bjorklund2014listing, duraj2020equivalences}), the algorithm for triangle listing is required list up to $T$ triangles in the the whole graph, while $(\alpha, \beta, \gamma, \rho, t)$-All-Edges Triangle Listing asks to list up to $t$ triangles for each edge. 

We also define an unbalanced triangle listing problem when we have to list up to $T$ triangles globally. 

\begin{definition}[$(\alpha, \beta, \gamma, \rho, T)$ Triangle Listing]
\label{def:triangle-listing}

Given a tripartite graph $G$ whose vertex set is $A \cup B \cup C$, such that $|A| = n^\alpha, |B| = n^\beta, |C| = n^\gamma$. Let $X, Y \in \{A, B, C\}$ be two different parts of the graph. For any $v \in X$, $\deg(v, Y) \le O(n^{-\rho} |Y|)$. The $(\alpha, \beta, \gamma, \rho, T)$ Triangle Listing problem asks to list all triangles in $G$ if there are fewer than $T$ triangles or list $T$ distinct triangles in $G$ if there are at least $T$ triangles. 
\end{definition}

Triangle Listing and All-Edges Triangle Listing are strongly related problems. 
For instance, it can be shown that $(\alpha, \beta, \gamma, \rho, t)$-All-Edges Triangle Listing can be reduced to, up to polylogarithmic factors, $(\alpha, \beta, \gamma, \rho, n^{\alpha+\beta - \rho} t)$ Triangle Listing, by a reduction similar to the reduction for Theorem 15 in \cite{duraj2020equivalences}. This reduction means that if we have hardness for $(\alpha, \beta, \gamma, \rho, t)$-All-Edges Triangle Listing, then we also have hardness for $(\alpha, \beta, \gamma, \rho, T)$ Triangle Listing when $T = n^{\alpha+\beta - \rho} t$. However, this argument won't work when $T < n^{\alpha+\beta - \rho}$, since it would require us to set $t < 1$, which doesn't make sense. Therefore, to circumvent this difficulty, we will directly reduce Exact Triangle to both Triangle Listing and All-Edges Triangle Listing.

The triangle listing problems require us to actually list triangles for some edge $e$. However, many problems we consider, including All-Edges Sparse Triangle and AE-Mono$\Delta$, only require the algorithms to output whether some triangle exists containing edge $e$.  Thus, in order to reduce to these problems, we define the following version of triangle detection.  

\begin{definition}[$(\alpha, \beta, \gamma, \rho)$-All-Edges Sparse Triangle]

Given a tripartite graph $G$ whose vertex set is $A \cup B \cup C$, such that $|A| = n^\alpha, |B| = n^\beta, |C| = n^\gamma$. Let $X, Y \in \{A, B, C\}$ be two different parts of the graph. For any $v \in X$, $\deg(v, Y) \le O(n^{-\rho} |Y|)$. The $(\alpha, \beta, \gamma, \rho)$-All-Edges Sparse Triangle problem asks to determine, for each $e \in E \cap (A \times B)$, whether $e$ is in a triangle or not. 
\end{definition}

Now we are ready to present the reduction from Exact Triangle to triangle listing problems. 

%\todo{maybe mention exact triangle == zero triangle in in the intro? }

\begin{theorem}
\label{thm:exact-to-general-triangle-listing}
Fix constants $0 \le \alpha, \beta, \gamma \le 1, \rho < \min\{\alpha, \beta, \gamma\}$.
There exists an $\tilde{O}(n^{3-\min\{\alpha, \beta, \gamma\} + \rho})$ time randomized reduction from a Zero Triangle instance with $n$ vertices to  $\tilde{O}(n^{3-\alpha-\beta-\gamma + 2\rho})$ instances of  $(\alpha, \beta, \gamma, \rho, 900 n^{\gamma - 2\rho} + 1)$-All-Edges Triangle Listing. Similarly, there is also an $\tilde{O}(n^{3-\min\{\alpha, \beta, \gamma\} + \rho})$ time randomized reduction from a Zero Triangle instance with $n$ vertices to  $\tilde{O}(n^{3-\alpha-\beta-\gamma+2\rho})$ instances of  $(\alpha, \beta, \gamma, \rho, 8100 n^{\alpha+\beta+\gamma - 3\rho} + 1)$ Triangle Listing.
\end{theorem}

\begin{proof}

We will first provide the reduction to All-Edges Triangle Listing. The reduction to Triangle Listing can be obtained via slight modifications. 

\textbf{Step 1: }

Fix a Zero Triangle instance $G$. We can randomly assign its vertices to one of three colors, and only keep edges whose two endpoints have different colors. 
 If we repeat $\Theta(\log n)$ times, a zero triangle in $G$ will remain at least one of the graphs. Thus it suffices to solve Zero Triangle on a tripartite graph.

Suppose $G$ is a tripartite graph with vertex parts $A, B, C$. First, we split vertex parts $A, B, C$ to parts of size $n^\alpha, n^\beta, n^\gamma$ respectively. We could enumerate all $n^{3-\alpha-\beta-\gamma}$ triples of parts, and detect zero triangle within each triple of parts. In the remainder of the reduction, it suffices to reduce each individual unbalanced zero triangle instance of vertex set sizes $|A|=n^\alpha, |B|=n^\beta, |C|=n^\gamma$ to
$(\alpha, \beta, \gamma, \rho, t)$-All-Edges Triangle Listing instances. 

\textbf{Step 2: }

We assume the edge weights $w(\cdot, \cdot)$ in the Zero Triangle instance are integers whose absolute values are bounded by $n^k$ for some constant $k \ge 1$. We pick an arbitrary prime $p$ that is between $100n^k$ and $Dn^k \log n$ for large enough constant $D$. By the Prime Number Theorem, a random integer in this range is a prime with probability $\Omega(1/\log n)$, so it takes $\tilde{O}(1)$ time to find such a prime by randomly picking integers in this range and test its primality. After we determine $p$, we can regard all the weights of the graph as elements in $\mathbb{F}_p$ by taking the weight of every edge modulo $p$. Since the weight of each triangle is in $[-3n^k, 3n^k]$ while $p \ge 100n^k$, the set of zero triangles with respect to the new weights stays the same. 

% Take a random prime $p \le C n^c \log^2 n$ for some large enough constant $C$ so that there are $100(k+1)n \log n$ primes in the interval (\todo{possible by PNT}). 
% For each edge $e$ with weight $w(e)$ in the graph, we set its new weight to be $w(e) \bmod {p} \in \mathbb{F}_p$.  For any edge $e \in E \cap (A \times B)$, if it is in a zero triangle in the old graph, then clearly, it is still in a zero triangle.  Otherwise, consider all $n^c$ triangles the edge $e$ is in. The weight of each of these triangles is an integer in $[-3n^k, 3n^k] \setminus \{0\}$, so each weight has at most $\log(3n^k) \le (k+1) \log n$ distinct prime factors (\todo{when $n$ sufficiently large}). Therefore, a  random prime $p \le C n^c \log^2 n$ is a factor of any of these $n^c$ triangles with probability at most $0.01$. 

\textbf{Step 3: }

Let $x \in \mathbb{F}_p$ be a random element from $\mathbb{F}_p$, and let $y_v \in \mathbb{F}_p$ be random elements from $\mathbb{F}_p$ for every $v \in A \cup B \cup C$. 
As illustrated in Figure~\ref{fig:random_weight},
for every edge $e = (a, b) \in E \cap (A \times B)$, we set its new weight to $w'(a,b) = x w(e) - y_b + y_a$; for every edge $e = (a, c) \in E \cap (A \times C)$, we set its new weight to $w'(a, c) = x w(e) - y_a + y_c$; for every edge $e = (b, c) \in E \cap (B \times C)$, we set its new weight to $w'(b, c) = x w(e) - y_c + y_b$. Let the graph with the new weights be $G'$.

\begin{figure}[ht]
\label{fig:randomizing-weight}
\centering
    \begin{tikzpicture}
    [v/.style={circle,draw,inner sep=0pt,minimum size=10pt}]
    \node[v] (1){$c$};
    \node[v] (2) [below left = 3cm and 1.73cm of 1]{$a$};
     \node[v] (3) [below right = 3cm and 1.73cm of 1]{$b$};
    
   \path (1) edge [left] node {$x \cdot w(a, c) - y_a + y_c$} (2);
  \path (1) edge [right] node {$x \cdot w(b, c) - y_c + y_b$} (3);
  \path (2) edge [below] node {$x \cdot w(a, b) - y_b + y_a$} (3);
    \end{tikzpicture}
    \caption{``randomizing" the weights. }
    \label{fig:random_weight}
\end{figure}
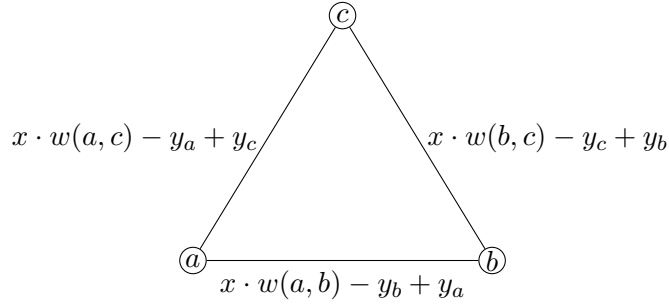
Clearly, as long as $x \ne 0$, the set of zero triangles with weights $w'(\cdot, \cdot)$ is exactly the same as the set of zero triangles with weights $w(\cdot, \cdot)$. Thus, false positives occur with probability $\frac{1}{p} \le \frac{1}{100n^k} \le 0.01$.

\textbf{Step 4: }

We split $\mathbb{F}_p$ into $n^\rho$ contiguous ranges $L_1, \ldots, L_{n^\rho}$, so that each range has size between $\lfloor p/n^\rho \rfloor$ and $\lceil p/n^\rho \rceil$.
We enumerate every triple of $i, j, k$ such that $0 \in L_i + L_j + L_k$. For every pair of $i, j$, the size of $L_i + L_j$ is $O(p/n^\rho)$. In order for $0 \in L_i + L_j + L_k$, we need $L_k \cap -(L_i + L_j) \ne \emptyset$. Since each $L_k$ has size $\Theta(p/n^\rho)$, there can be at most $O(1)$ ranges $L_k$ that intersect with $-(L_i + L_j)$. 
Thus, the total number of such triples is $O(n^{2\rho})$.

For each triple $(i, j, k)$, we consider a subset of edges $E_{i,j,k}$ defined as 
$$\{e \in E \cap (A \times C): w'(e) \in L_i\} \cup \{e \in E \cap (B \times C): w'(e) \in L_j\} \cup \{e \in E \cap (A \times B): w'(e) \in L_k\}.$$
Let $G_{i,j,k} = (A \cup B \cup C, E_{i, j, k})$ be the subgraph of $G'$ with edge set $E_{i, j, k}$. Clearly, if graph $G'$ has a zero triangle, one of $G_{i,j,k}$ will have a zero triangle, and vice versa. 

Now we change $G_{i,j,k}$ so that the degree of every vertex is bounded. For each $v \in A$, if $\deg(v, |B|) > 100 |B| n^{-\rho} + 200$ or $\deg(v, |C|) > 100 |C| n^{-\rho}+ 200$, we remove the vertex $v$ from graph $G_{i, j, k}$. We similarly handle vertices in parts $B$ and $C$ that have large degrees.

Finally, we use an algorithm for $(\alpha, \beta, \gamma, \rho, t)$-All-Edges Triangle Listing for $t = 900 n^{c-2\rho} + 1$ on graph $G_{i, j, k}$ to list up to $900 n^{c-2\rho} + 1$ triangle for each edge $e \in E_{i, j, k} \cap (A \times B)$. If any of these listed triangles is a zero triangle in the original graph $G$, we return YES to the Zero Triangle instance.

\textbf{Step 5:}

We repeat the previous steps $100 \log n$ times. If no zero triangle is found in any of these $100 \log n$ tries, we return NO to the Zero Triangle instance. 

\textbf{Analysis}

It should be clear why Step $1$ through Step $3$ works via the in-text explanations. Now we prove why Step $4$ and Step $5$ work.

\begin{claim}
\label{cl:dont-remove}
Fix any zero triangle $(a, b, c)$ in $G_{i, j, k}$ where $a \in A, b \in B, c \in C$. With probability at least $0.94$, none of $a, b, c$ will be removed in Step $4$ due to having a large degree. 
\end{claim}
\begin{proof}
First consider $\deg(a, B)$. For any $b' \in B \setminus \{b\}$, $w'(a, b') = x \cdot w(a, b') - y_{b'} + y_a$. Conditioned on the fact that $(a, b), (b, c), (c, a) \in E_{i, j, k}$, $y_{b'}$ is a completely new random variable. Therefore, $\Pr[(x, b') \in E_{i, j, k}] = \Pr[w(a, b') - y_{b'} + y_a \in L_k] = \frac{|L_k|}{p} \le \frac{\lceil pn^{-\rho}\rceil}{p} \le n^{-\rho}+\frac{1}{p}.$ Therefore, the expected value of $\deg(a, B)$ can be written as 
$$\mathbb{E}\left[\deg(a, B)\right] = 1 + \mathbb{E}\left[ \sum_{b' \ne b} [(a, b') \in E_{i, j, k}]\right] \le 1+ |B|n^{-\rho} + \frac{|B|}{p} \le 2 + |B|n^{-\rho}.$$
Thus, by Markov's inequality, $\Pr\left[\deg(a, B) > 200 + 100 |B|n^{-\rho}\right] \le 0.01$. We can apply the same argument to $\deg(a, C), \deg(b, A), \deg(b, C), \deg(c, A)$ and $\deg(c, B)$ and take a union bound. Thus, with probability at least $0.94$, all of these degrees will be small enough so that none of $a, b, c$ are removed. 

\end{proof}

We also need to show that listing $900n^{\gamma-2\rho}  + 1$ triangles for each edge will be enough, i.e. there are not too many false positives for each edge $e \in E_{i, j, k} \cap (A \times B)$.

\begin{claim}
\label{cl:few-false-positive}
Fix any zero triangle $(a, b, c) $ in $G_{i, j, k}$ where $a \in A, b \in B$, $c \in C$. The number of vertices $c'$ such that 
\begin{enumerate}[label=\arabic*)]
\item $(a, c'), (b, c') \in E_{i,j,k}$, and
\item $w(a, b) + w(b, c') + w(c', a) \ne 0$ (Recall that $w(\cdot, \cdot)$ is viewed as elements in $\mathbb{F}_p$, so all operations are modulo $p$).
\end{enumerate}
is at most $900 n^{\gamma-2\rho}$ with probability at least $0.99$.
\end{claim}
\begin{proof}
Let $c'$ be any $c' \in C$ such that $w(a, b) + w(b, c') + w(c', a) \ne 0$.
If $(a, c'), (b, c') \in E_{i,j,k}$, then $w'(a, c'), w'(a, c) \in L_i$ and $w'(b, c'), w'(b, c) \in L_j$. Since $L_i$ and $L_j$ are both ranges of size at most $\lceil p n^{-\rho} \rceil$, we must have 
\[
  \begin{cases} 
   w'(a, c) - w'(a, c') \in \left[- p n^{-\rho} ,   pn^{-\rho}  \right] \\
   w'(b, c) - w'(b, c') \in \left[-  pn^{-\rho} ,  pn^{-\rho}  \right].
  \end{cases}
\]
We can expand out the definition of $w'$  to get 
\begin{equation}
    \label{eq:eq1}
  \begin{cases} 
  x \cdot (w(a, c) - w(a, c')) + (y_c - y_{c'})\in [- p n^{-\rho}, pn^{-\rho}]  \\
  x \cdot (w(b, c) - w(b, c')) + (y_{c'} - y_{c}) \in [- p n^{-\rho}, pn^{-\rho}].
  \end{cases}
\end{equation}
Each of the two values in Equation~(\ref{eq:eq1}) is clearly uniformly at random. To show these two values are independent, we consider the sum of these two values, which is $x \cdot (w(a, c) + w(b, c) - w(a, c') - w(b, c'))$. 
Since $w(a, b) + w(b, c') + w(c', a) \ne 0$, while $w(a, b) + w(b, c) + w(c, a) = 0$, we have 
 $w(a, c) + w(b, c) \ne w(a, c') + w(b, c')$. Thus, the sum of the two values in Equation~(\ref{eq:eq1}) is the product of $x$ with a nonzero value. Thus, the sum of these two values is also a uniformly random variable. Conditioned on the sum, $x \cdot (w(a, c) - w(a, c')) + (y_c - y_{c'})$ is also uniformly at random, since the sum does not contain the $y_{c'}$ term. Thus the two values in Equation~(\ref{eq:eq1}) are independent. Therefore, the probability that Equation~(\ref{eq:eq1}) happens is at most $(2n^{-\rho}+\frac{1}{p})^{2} \le 9n^{-2\rho}$. 

Summing over all $c' \in C$, the expected number of $c'$ satisfying conditions 1) and 2) is at most $9n^{\gamma-2\rho} $. Thus, the probability that the number of such $c'$ exceeds $900n^{\gamma-2\rho} $ is at most $0.01$. 

\end{proof}

If an edge $(a, b)$ is in a zero triangle $(a, b, c) $ in the original graph $G$, then this zero triangle is preserved in one instance $G_{i, j, k}$ before removing any vertices. Then we 
will report a triangle containing edge $(a, b)$ as long as 
\begin{enumerate}[label=\arabic*)]
\item we don't remove any of $a, b, c$ in the vertex removal process in Step $4$ (which happens with probability at least $0.94$ by Claim~\ref{cl:dont-remove});
\item the number of nonzero triangles containing $(a, b)$ in $G_{i, j, k}$ is at most $900n^{\gamma-2\rho}$(which happens with probability at least $0.99$ by Claim~\ref{cl:few-false-positive}).
\end{enumerate}
Therefore, by union bound, each iteration reports at least one zero triangle with constant probability if the original graph has a zero triangle. Thus, repeating the iterations for $O(\log n)$ time suffices. 

\textbf{Running Time:}

We split the $n$-node graph  to $n^{3-\alpha - \beta - \gamma}$ unbalanced graphs in Step $1$. For each unbalanced graph, we reduce it to $n^{2\rho}$ instances of $(\alpha, \beta, \gamma, \rho, 900 n^{\gamma - 2\rho} + 1)$-All-Edges Triangle Listing. 
The running time of the reduction is linear (up to poly-logarithmic factors) to the total input size of all the triangle listing instances. Thus, the running time is 
$$\tilde{O}(n^{3-\alpha-\beta-\gamma + 2\rho} \cdot (n^{\alpha+\beta-\rho} + n^{\beta+\gamma-\rho} + n^{\alpha+\gamma-\rho})) = \tilde{O}(n^{3-\min\{\alpha,\beta, \gamma\} + \rho}).$$

\textbf{Proof of the reduction from Zero Triangle to Triangle Listing}

The reduction is largely the same as the previous reduction. The only difference in the reduction is that now in Step 4, we use the $(\alpha, \beta, \gamma, \rho, T)$ Triangle Listing algorithm on graph $G_{i, j, k}$ to list up to $T=8100n^{\alpha+\beta+\gamma - 3\rho} + 1$ triangles, and test whether any of these triangles is a zero triangle. The correctness analysis of this reduction requires a more careful analysis of the expected number of triangles in $G_{i, j, k}$.

\begin{claim}
\label{cl:global-number-of-triangle}
Fix any triple  $(i, j, k)$ so that $G_{i, j, k}$ contains at least one zero triangle with respect to weight $w(\cdot, \cdot)$. With probability at least $0.99$, the number of triangles in $G_{i, j,k}$ that are not zero triangles in the original graph $G$ is at most $8100n^{\alpha+\beta+\gamma -3\rho}$. 
\end{claim}
\begin{proof}

Let $(a, b, c) \in G_{i, j, k}$ be one arbitrary zero triangle in $G_{i, j, k}$.
We analyze the probability that each nonzero triangle $(a', b', c') \in G$ belongs to $G_{i, j, k}$.
In Claim~\ref{cl:few-false-positive}, we already analyzed the case for triangles with $a=a'$ and $b=b'$, in which case the expected number of $(a', b', c')$ inside $G_{i, j, k}$ is at most $9n^{\gamma-2\rho}$. We can similarly bound the expected number of triangles $(a', b', c')$ in $G_{i, j, k}$ that share two vertices with triangle $(a, b, c)$. 
The expected number of such triangles can be shown to be at most $9(n^{\alpha} + n^{\beta} + n^{\gamma})n^{-2\rho}$.

Thus, it suffices to analyze the remaining cases when $(a, b, c)$ and $(a', b', c')$ share one or zero common vertices. 
First, consider triangles $(a', b', c')$ that share one vertex with $(a, b, c)$. Without loss of generality, assume such triangles have the form $(a, b', c')$ for some $b \ne b', c \ne c'$. In order for $(a, b, c)$ and $(a, b', c')$ happen to be in the same edge set $E_{i, j,k}$, we must have 
\[
  \begin{cases} 
  x \cdot (w(a, b) - w(a, b'))-y_b+y_{b'}\in [-pn^{-\rho}, pn^{-\rho}]\\
%   w'(a, b) - w'(a, b') \in [-\rho p, \rho p] \\
x \cdot (w(a, c) - w(a, c'))+y_c-y_{c'}\in [-pn^{-\rho}, pn^{-\rho}]\\
%   w'(a, c) - w'(a, c') \in [-\rho p, \rho p]\\
x\cdot(w(b,c)-w(b',c')) -y_c+y_b + y_{c'}-y_{b'} \in [-pn^{-\rho}, pn^{-\rho}].\\
%   w'(b, c) - w'(b',c') \in [-\rho p, \rho p]\\
  \end{cases}
\]
Let $X_{a, b}, X_{a, c}, X_{b, c}$ be the random variables denoting the three expressions in the above condition respectively. 
We will show that $X_{a, b}, X_{a, c}, X_{b, c}$ are independent. First, we analyze the sum of these three random variables. Consider $X_{a, b} + X_{a, c} + X_{b, c}$, which equals $x \cdot (w(a, b) + w(a, c) + w(b, c) - w(a,b') - w(a, c') - w(b',c'))$. 
Note that $(a, b', c')$ is not a zero triangle, while $(a, b, c)$ is. Thus, the sum is the product of $x$ and a nonzero value, so the result is a uniformly random value. $X_{a, b}$ has an additive term $y_{b'}$ that is independent of $X_{a, b} + X_{a, c} + X_{b, c}$, so $X_{a, b}$ is independent of $X_{a, b} + X_{a, c} + X_{b, c}$. Similarly, $X_{a, c}$ has an additive term $y_{c'}$, which is independent of $(X_{a, b} + X_{a, c} + X_{b, c}, X_{a, b})$. Thus, $X_{a, b} + X_{a, c} + X_{b, c}$, $X_{a, b}$ and $X_{a, c}$ are independent, which implies $X_{a, b}, X_{a, c}$ and $X_{b, c}$ are independent. 

The probability that $X_{a, b} \in [-p n^{-\rho}, p n^{-\rho}]$ is at most $\frac{2p n^{-\rho} + 1}{p} \le 3n^{-\rho}$. Similarly, the probability that $X_{a, c}, X_{b, c} \in [-p n^{-\rho}, p n^{-\rho}]$ are both at most $2n^{-\rho} + 1/p \le 3n^{-\rho}$. Since these three random variables are independent, the probability that all of the three are in $[-p n^{-\rho}, p n^{-\rho}]$ is at most $27n^{-3\rho}$. This means that $(a, b, c)$ and $(a, b', c')$ are in the same edge set $E_{i,j, k}$ with probability at most $27n^{-3\rho}$. More generally, if $(a', b', c')$ share exactly one common vertex with $(a, b, c)$, it will be in $E_{i, j, k}$ with probability at most $27n^{-3\rho}$. 

Triangles that share zero vertices with $(a, b, c)$ can be analyzed similarly, and each of them is in $E_{i, j, k}$ with probability at most $27 n^{-3\rho}$ as well. 

Thus, the expected number of nonzero triangles in $G_{i,j,k}$ is at most $2 \cdot 27n^{\alpha+\beta+\gamma-3\rho} + 9 (n^{\alpha} + n^{\beta} + n^{\gamma})n^{-2\rho}$. By Markov's inequality, with probability at least $0.99$, the number of nonzero triangles in $G_{i, j, k}$ is at most $5400n^{\alpha+\beta+\gamma-3\rho} + 900  (n^{\alpha} + n^{\beta} + n^{\gamma})n^{-2\rho}$. Since $\rho < \min\{\alpha, \beta, \gamma\} $, we have $n^{\alpha-2\rho} , n^{\beta-2\rho} , n^{\gamma-2\rho} \le n^{\alpha+\beta+\gamma-3\rho}$. Therefore, $5400n^{\alpha+\beta+\gamma-3\rho} + 900  (n^{\alpha} + n^{\beta} + n^{\gamma})n^{-2\rho} \le 8100n^{\alpha+\beta+\gamma-3\rho} $. 
\end{proof}

By Claim~\ref{cl:global-number-of-triangle}, we know that is suffices to list $8100n^{\alpha+\beta+\gamma-3\rho} + 1$ triangles in each graph $G_{i, j, k}$.

\end{proof}

Theorem~\ref{thm:exact-to-general-triangle-listing} shows hardness for listing triangles in some special graphs. To show the hardness for detecting triangles, we still need a reduction from triangle listing to triangle detection. Theorem $15$ in \cite{duraj2020equivalences} is such a reduction that reduces listing $O(1)$ triangles for each edge to detecting whether each edge is in a triangle; however, that reduction changes the structure of the graph. Specifically, it does not necessarily reduce an $(\alpha,\beta, \gamma, \rho, O(1))$-All-Edges Triangle Listing instance to $(\alpha,\beta, \gamma, \rho)$-All-Edges Sparse Triangle instances. Thus, we give a new structure-preserving reduction from triangle listing to all edge triangle detection. The reduction adapts the techniques for finding the witnesses of Boolean matrix multiplication~\cite{AlonGMN92, vnotes2}.
%\todo{Citation for this? }.

\begin{proposition}
\label{prop:triangle-listing-to-detection}
Let $\alpha, \beta, \gamma$ be any positive constants and let $\rho \le \min\{\alpha, \beta, \gamma\}$. 
There exists an $\tilde{O}((n^{\alpha+\beta}+n^{\beta+\gamma}+n^{\gamma+\alpha})n^{-\rho})$ time randomized reduction from an $(\alpha, \beta, \gamma, \rho, O(1))$-All-Edges Triangle Listing instance to $\tilde{O}(1)$ instances of $(\alpha, \beta, \gamma, \rho)$-All-Edges Sparse Triangle. 
\end{proposition}
\begin{proof}
The reduction proceeds in two parts. We define an intermediate problem called $(\alpha, \beta, \gamma, \rho)$-All-Edges Unique Triangle Listing, where we are given a graph that shares the same structure as $(\alpha, \beta, \gamma, \rho)$-All-Edges Sparse Triangle instances, and
we seek an algorithm that outputs a triangle for every edge $e \in E \cap (A \times B)$ only if there is a unique triangle containing edge $e$; otherwise, the algorithm can output $0$ for edge $e$. 

In the first part, we show a reduction from $(\alpha, \beta, \gamma, \rho, O(1))$-All-Edges Triangle Listing to $(\alpha, \beta, \gamma, \rho)$-All-Edges Unique Triangle Listing. In the second part, we show a reduction from $(\alpha, \beta, \gamma, \rho)$-All-Edges Unique Triangle Listing to $(\alpha, \beta, \gamma, \rho)$-All-Edges Sparse Triangle. 

Fix an $(\alpha, \beta, \gamma, \rho, k)$-All-Edges Triangle Listing instance $G$ for some constant $k$. For each integer $\ell$ from $1$ to $c \log n$, we perform a stage. In each stage $\ell$, we repeat the following iterations for $\Theta(k^2 \log n)$ times. 
In each iteration, we create a new graph $G'$ that contains parts $A$ and $B$ and a subset $C'$ of part $C$. We obtain $C'$ by independently keeping every vertex $c \in C$ with probability $\frac{1}{2^\ell}$. Then $G'$ is the induced subgraph of $G$ with vertex set $A\cup B \cup C'$. 
We run an algorithm for $(\alpha, \beta, \gamma, \rho)$-All-Edges Unique Triangle Listing on $G'$ to list at most one triangle for each edge. If the algorithm lists a triangle for edge $(a, b)$, we add this triangle to a set $S_{(a, b)}$ that contains a list of found triangles containing edge $(a, b)$. After all the rounds, we output up to $k$ distinct triangles from each $S_{(a, b)}$. 

To show the correctness for this algorithm, we show that if the actual number of triangles containing edge $(a, b)$ is $\Delta$ for some  $\Delta \in [2^{\ell - 1},  2^{\ell}]$, then in stage $\ell$ we will list all or up to $k$ triangles containing $(a, b)$ with high probability. In every iteration, we pick a random subset $C' \subseteq C$. For every $(a, b)$, suppose $|S_{(a, b)}| < k$ and there are more triangles containing $(a, b)$ that have not been found. 
The probability that we keep a unique triangle not in $S_{(a, b)}$ for edge $(a, b)$ is at least $(\Delta-|S_{(a, b)}|) \cdot \frac{1}{2^{\ell}} (1-\frac{1}{2^{\ell}})^{\Delta - 1} = \Omega(1/k)$. Thus, after every $\Theta(k \log n)$ iterations, we will find a new triangle for edge $(a, b)$ with high probability if $|S_{(a, b)}| < \min\{k, \Delta\}$. Therefore, we need $\Theta(k^2 \log n)$ rounds in total.

Now we show the second part of the reduction, which reduces from $(\alpha, \beta, \gamma, \rho)$-All-Edges Unique Triangle Listing to $(\alpha, \beta, \gamma, \rho)$-All-Edges Sparse Triangle. Let $G$ be a graph on which we want to solve $(\alpha, \beta, \gamma, \rho)$-All-Edges Unique Triangle Listing.
For every $1 \le i \le \gamma \log n$, we create a graph $G^i$ that contains all vertices of $A$ and $B$, but only those vertices from $C$ whose $i$-th bit in its binary representation is $1$. Then we run an algorithm for $(\alpha, \beta, \gamma, \rho)$-All-Edges Sparse Triangle on graph $G^i$.
Suppose $(a, b)$ is in a unique triangle $(a, b, c)$. Then if $(a, b)$ is in a triangle in $G^i$, the $i$-th bit of $c$ must be $1$; otherwise, the $i$-th bit of $c$ must be $0$. Therefore, we will be able to determine $c$ after all the iterations. If $(a, b)$ is not in a unique triangle, then the value $c$ we determine might not form a triangle. In this case, we can determine that $(a, b, c)$ does not form a triangle and output $0$ for edge $(a, b)$.

\end{proof}

\begin{corollary}
\label{cor:exact-to-sparse}
There exists a reduction from Exact Triangle to $\tilde{O}(n)$ instances of All-Edges Sparse Triangle of $O(n^{1.5})$ edges. Thus, assuming the Exact Triangle hypothesis, there is no $O(m^{4/3-\epsilon})$ time algorithm for All-Edges Sparse Triangle for $\epsilon > 0$. 
\end{corollary}
\begin{proof}
If we set $\alpha = \beta = \gamma = 1$ and $\rho=0.5$ in Theorem~\ref{thm:exact-to-general-triangle-listing}, we get an $\tilde{O}(n^{2.5})$ time reduction from Exact Triangle to $\tilde{O}(n)$ instances of $(1, 1, 1, 0.5, O(1))$-All-Edges Triangle Listing. By Proposition~\ref{prop:triangle-listing-to-detection}, these instances further reduce to $\tilde{O}(n)$ instances of $(1, 1, 1, 0.5)$-All-Edges Sparse Triangle. These instances can be solved by an algorithm for All-Edges  Sparse Triangle with $O(n^{1.5})$ edges. 

Thus, if there is an $O(m^{4/3-\epsilon})$ time algorithm for All-Edges Sparse Triangle, we can use it to solve Exact Triangle in $\tilde{O}(n^{2.5} + n \cdot (n^{1.5})^{4/3 - \epsilon}) = \tilde{O}(n^{2.5}+n^{3 - 1.5\epsilon})$ time, breaking the Exact Triangle hypothesis. 
\end{proof}

\begin{corollary}
\label{cor:mono-triangle-hard}
Assuming the Exact Triangle hypothesis, there is no $O(n^{2.5 -\epsilon})$ algorithm for AE-Mono$\Delta$  on $n$-node graphs for $\epsilon > 0$.
\end{corollary}
\begin{proof}
Combining Corollary~\ref{cor:exact-to-sparse} and Theorem~\ref{thm:sparse_to_mono}, we get an $\tilde{O}(n^{2.5})$ time reduction from Exact Triangle of size $n$ to $\tilde{O}(\sqrt{n})$ instances of AE-Mono$\Delta$ of size $O(n)$. Thus, if there is an $O(n^{2.5-\epsilon})$ time algorithm for AE-Mono$\Delta$, we can use it to solve Exact Triangle in $\tilde{O}(n^{2.5}+\sqrt{n} \cdot n^{2.5-\epsilon}) = \tilde{O}(n^{2.5}+n^{3-\epsilon})$ time, breaking the Exact Triangle hypothesis.  
\end{proof}

\begin{corollary}
Let $\mathbb{A}$ be an algorithm for All-Edges Sparse Triangle for  $n$-node graphs where every node has degree at most $d=n^\delta$ for some $0 < \delta \le 0.5$.
Assuming the Exact Triangle hypothesis, $\mathbb{A}$ cannot run in $O(n^{1-\epsilon}d^2)$ for $\epsilon > 0$. 
\end{corollary}
\begin{proof}
We set $\alpha=\beta=\gamma=1$ and $\rho =1-\delta$ in Theorem~\ref{thm:exact-to-general-triangle-listing}. This gives an $\tilde{O}(n^{3-\delta})$ time reduction from Exact Triangle to $\tilde{O}(n^{2-2\delta})$ instances of $(1, 1, 1, \rho, O(1))$-All-Edges Triangle Listing. Thus, each of these instances requires $n^{1+2\delta - o(1)}$ time. By Proposition~\ref{prop:triangle-listing-to-detection}, there is a reduction from $(1, 1, 1, \rho, O(1))$-All-Edges Triangle Listing to $\tilde{O}(1)$ instances of $(1, 1, 1, \rho)$-All-Edges Sparse Triangle. Therefore, $(1, 1, 1, \rho)$-All-Edges Sparse Triangle -- All-Edges Sparse Triangle in graphs with maximum degree $n^\delta$ -- also requires $n^{1+2\delta - o(1)}$ time.
\end{proof}

% The input to the SetDisjointness problem is an universe $U$, a family $\mathcal{F} \subseteq 2^U$ of subsets of $U$, and $q$ pairs of queries $(S, S') \in \mathcal{F} \times \mathcal{F}$. For each query $(S, S')$, we want to determine if $S \cap S'$ is empty or not. 

\begin{corollary}
\label{cor:setdisjointness}
For any constant $0 < \theta < 1$, let $\mathbb{A}$ be an algorithm for offline SetDisjointness where $|U| = \Theta(N^{2-2\theta})$, $|\mathcal{F}| = \Theta(N)$, each set in $\mathcal{F}$ has at most $O(N^{1-\theta})$ elements from $U$, and $q=\Theta(N^{1+\theta})$. Assuming the Exact Triangle hypothesis, $\mathbb{A}$ cannot run in $O(N^{2-\epsilon})$ for $\epsilon > 0$. 
\end{corollary}
\begin{proof}
Set $\alpha = \beta=0.5$, $\gamma=1-\theta$ and $\rho =0.5 - \theta / 2$ in Theorem~\ref{thm:exact-to-general-triangle-listing}. Thus there is a subcubic time reduction from Exact Triangle to $\tilde{O}(n^{2})$ instances of $(0.5, 0.5, 1 - \theta, 0.5 - \theta / 2, O(1))$-All-Edges Triangle Listing. Thus, assuming the Exact Triangle hypothesis, $(0.5, 0.5, 1 - \theta, 0.5 - \theta / 2, O(1))$-All-Edges Triangle Listing requires $n^{1-o(1)}$ time. By Proposition~\ref{prop:triangle-listing-to-detection}, $(0.5, 0.5, 1 - \theta, 0.5 - \theta / 2)$-All-Edges Sparse Triangle also requires $n^{1-o(1)}$ time. As realized in previous works (e.g. \cite{kopelowitz2016higher}), All-Edges Sparse Triangle can be solved using SetDisjointness. Using the language of the $(0.5, 0.5, 1 - \theta, 0.5 - \theta / 2)$-All-Edges Sparse Triangle problem, we can set the universe $U$ to the vertex set $C$, and set the family $\mathcal{F}$ to $A \sqcup B$. We add an element $u \in U$ to $S \in \mathcal{F}$ if the corresponding vertex of $S$ has an edge with the corresponding vertex of $u$. Finally, for any edge $(a, b) \in A\times B$, we add a query for the two sets corresponding to vertices $a$ and $b$. Then clearly, $(a, b)$ is in a triangle if and only if their corresponding sets intersect. 
Setting $N = \sqrt{n}$ finishes the proof.  
\end{proof}

% Now we define another problem investigated in \cite{kopelowitz2016higher}: SetIntersection. 
% The input to the SetIntersection problem is the same as the input to the SetDisjointness problem, but for SetIntersection, we are also required to output elements in the intersection of each query $(S, S')$. However, we are only required to output a limited number of such elements. 

\begin{corollary}
For any constants $0 \le \theta < 1$ and $0 < \delta$, let $\mathbb{A}$ be an algorithm for offline SetIntersection where $|U| = \Theta(N^{1+\delta - \theta})$, $|\mathcal{F}| = \Theta(\sqrt{N^{1+\delta + \theta}})$, each set in $\mathcal{F}$ has at most $O(N^{1-\theta})$ elements from $U$, $q=\Theta(N^{1+\theta})$, and $T=O(N^{2-\delta})$. Assuming the Exact Triangle hypothesis, $\mathbb{A}$ cannot run in $O(N^{2-\epsilon})$ for $\epsilon > 0$. 
\end{corollary}
\begin{proof}
If $\delta - \theta \ge 1$, the lower bound is trivially true because the input size is $|\mathcal{F}| \cdot N^{1-\theta} = N^{1.5 + 0.5 \delta - 0.5 \theta} \ge N^2$. Thus, we assume $\delta - \theta < 1$. 

Set $\alpha=\beta=\frac{1}{2}+\frac{\theta}{2+2\delta}, \gamma=1-\frac{\theta}{1+\delta}$, and $\rho =\frac{\delta}{1+\delta}$ in Theorem~\ref{thm:exact-to-general-triangle-listing}. This yields a reduction from Exact Triangle to $\tilde{O}(n^{1+\frac{2\delta}{1+\delta}})$ instances of $(\frac{1}{2}+\frac{\theta}{2+2\delta},\frac{1}{2}+\frac{\theta}{2+2\delta},1-\frac{\theta}{1+\delta}, \frac{\delta}{1+\delta}, O(n^{2-\frac{3\delta}{1+\delta}}) )$ Triangle Listing. Assuming the Exact Triangle hypothesis, each triangle listing instance requires $n^{\frac{2}{1+\delta}-o(1)}$ time to compute. Similar to the fact that SetDisjointness can be used to solve All-Edges Sparse Triangle, SetIntersection can be used to solve Triangle Listing. Thus, 
setting $N=n^{\frac{1}{1+\delta}}$ finishes the proof. 
\end{proof}

\section{Reductions Between Intermediate Problems}
\label{sec:intermediate}

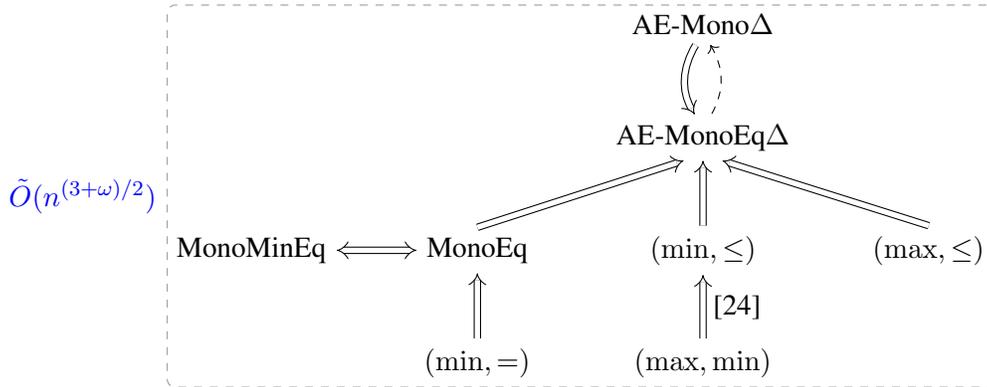
\begin{figure}[ht]
    \centering
    \begin{tikzpicture}
    
        \node at(0, 0)  [] (monoT){AE-Mono$\Delta$};
        \node at(0, -1.5)  [] (monoeqT){AE-MonoEq$\Delta$};
        \node at(-3, -3)  [] (monoeqP){MonoEq};
        \node at(-6, -3)  [] (monomineqP){MonoMinEq};
        \node at(0, -3)  [] (minle){$(\min, \le)$};
        \node at(3, -3)  [] (maxle){$(\max, \le)$};
        
        \node at(-3, -4.5)  [] (mineq){$(\min, =)$};
        \node at(0, -4.5)  [] (maxmin){$(\max, \min)$};

         \draw[<-,dashed, bend left] (monoT.290) to[]  node[] {} (monoeqT.70);
        \draw[-Implies, double distance=2pt, bend right] (monoT.250) to[]  node[] {} (monoeqT.110);
        
        \draw[-Implies,double distance=2pt]  (monoeqP.north) to[]  node[] {} (monoeqT.230);
        \draw[Implies-Implies,double distance=2pt]  (monoeqP.west) to[]  node[] {} (monomineqP.east);
        \draw[-Implies,double distance=2pt]  (minle.north) to[]  node[] {} (monoeqT.south);
        \draw[-Implies,double distance=2pt] (maxle.north) to[]  node[] {} (monoeqT.310);
        
        \draw[-Implies,double distance=2pt] (mineq.north) to[]  node[] {} (monoeqP.south);
        \draw[-Implies,double distance=2pt] (maxmin.north) to[]  node[right] {\cite{VassilevskaWY07}} (minle.south);
        
        \draw[opacity=0.4, dashed, rounded corners=3] (current bounding box.north east) -- (current bounding box.north west) -- (current bounding box.south west) -- (current bounding box.south east) -- cycle;
        
        \node at(current bounding box.west)  [anchor=east] (){\textcolor{blue}{$\tilde{O}(n^{(3+\omega)/2})$}};
        
    \end{tikzpicture}
    \caption{Main reductions in Section~\ref{sec:intermediate}. Double arrows represent that the running times before and after the reduction are the same up to poly-logarithmic factors. Dashed arrows represent reductions that hold only when $\omega > 2$. }
    \label{fig:AE_monoEq_triangle_reductions}
\end{figure}

We will use the following known facts about the exponent of rectangular matrix multiplication in this section. 
\begin{theorem}\cite{huang1998fast}
\label{thm:huang98}
For any $k > 1$ and for any integer $q \ge 3$, $$\omega(1, k, 1) \le \frac{\log\left((1+k)^{(1+k)} \left(\frac{2+q}{2+k}\right)^{2+k}\right)}{\log q}.$$
\end{theorem}

\begin{corollary}
\label{cor:rect-approach-inf}
For any $\delta > 0$, there exists a number $k \ge 3$ such that $\omega(1, k, 1) \le 1 + k + \delta$. 
\end{corollary}
\begin{proof}
Let $k \ge 3$ and use $q = k$ in Theorem~\ref{thm:huang98}. Thus, $w(1, k, 1) \le (1+k) \frac{\log(k+1)}{ \log k}$. Consider $(1+k) \frac{\log(k+1)}{\log k} - (1+k) = (1+k) \frac{\log(k+1) - \log(k)}{ \log k}$. Since $\log(k+1) -\log(k) = \int_{k}^{k+1} \frac{1}{x} dx < \int_k^{k+1} \frac{1}{k} dx = \frac{1}{k}$, we must have $(1+k) \frac{\log(k+1)}{\log k} - (1+k) \le \frac{1+k}{k \log k} $.
When $k$ is large enough, $ \frac{1+k}{k \log k} \le \delta$ and thus $\omega(1, k, 1) \le 1 + k + \delta$. 
\end{proof}

We will also use the following fact about the convexity of $\omega(1, x, 1)$ (see e.g. \cite{le2012faster}, \cite{lotti1983asymptotic}).
\begin{fact}
\label{fact:rect_MM_convex}
When $0 < p \le k \le q$, $\omega(1, k, 1) \le \frac{k-p}{q-p} \omega(1, q, 1) + \frac{q-k}{q-p} \omega(1, p, 1)$.
\end{fact}

% \begin{definition}[AE-Mono$\Delta$]
% Given a graph $G$, each edge has a color $c$ and a value $v$. We call a triangle good if all its three edges have the same color and at least two of its edges have the same value. For each edge $e \in E(G)$, determine if it is in a good triangle. 
% \end{definition}

Consider an AE-MonoEq$\Delta$ instance $G=(V, E)$. We can copy the vertex set of $G$ three times to part $I, J, K$, and then plant the edges of $G$ to $I \times J, J \times K$ and $I \times K$. Thus we only need to solve AE-MonoEq$\Delta$ on tripartite graphs. Say we want to report for every edge between parts $I$ and $J$, whether it is in a monochromatic equality triangle. Depending on where the two edges with the same values are, we can split the problem to  three cases shown in Figure~\ref{fig:monochromatic-equality-triangle-cases}.

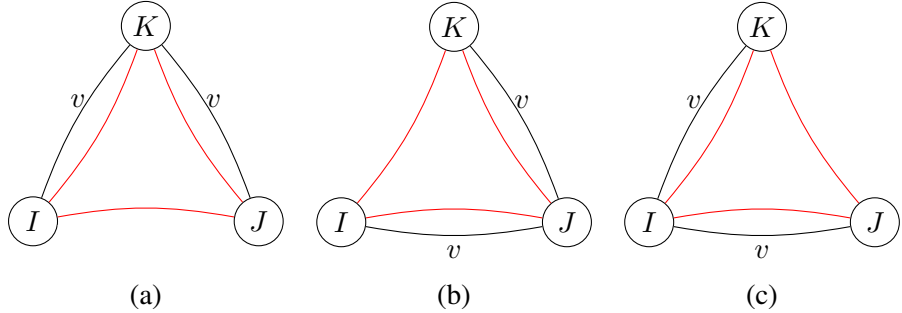
\begin{figure}[h]
\centering
    \begin{minipage}[t]{4cm}
    \centering
    \begin{tikzpicture}[
        state/.style ={ circle ,
        draw,black , text=black , inner sep=0pt, text width=6mm, align=center}]
        \node at(0,0)  [state] (I){$I$};
        \node at(3, 0)  [state] (J){$J$};
        \node at(1.5, 2.6)  [state] (K){$K$};
        \draw[-, red] (I) to[bend left=10]  node[] {} (J);
        % fake edge to keep things aligned
        \draw[-, white] (I) to[bend right=10]  node[below] {\textcolor{white}{v}} (J);
        \draw[-, red] (J) to[bend left=10]  node[] {} (K);
        \draw[-] (J) to[bend right=10]  node[above] {$v$} (K);
        \draw[-, red] (K) to[bend left=10]  node[] {} (I);
        \draw[-] (K) to[bend right=10]  node[above] {$v$} (I);
    \end{tikzpicture}
    (a)
    \end{minipage}
    \begin{minipage}[t]{4cm}
    \centering
    \begin{tikzpicture}[
        state/.style ={ circle ,
        draw,black , text=black , inner sep=0pt, text width=6mm, align=center}]
        \node at(0,0)  [state] (I){$I$};
        \node at(3, 0)  [state] (J){$J$};
        \node at(1.5, 2.6)  [state] (K){$K$};
        \draw[-, red] (I) to[bend left=10]  node[] {} (J);
        \draw[-] (I) to[bend right=10]  node[below] {$v$} (J);
        \draw[-, red] (J) to[bend left=10]  node[] {} (K);
        \draw[-] (J) to[bend right=10]  node[above] {$v$} (K);
        \draw[-, red] (K) to[bend left=10]  node[] {} (I);
        % \draw[-] (K) to[bend right=10]  node[above] {v} (I);
    \end{tikzpicture}
    (b)
    \end{minipage}
    \begin{minipage}[t]{4cm}
    \centering
    \begin{tikzpicture}[
        state/.style ={ circle ,
        draw,black , text=black , inner sep=0pt, text width=6mm, align=center}]
        \node at(0,0)  [state] (I){$I$};
        \node at(3, 0)  [state] (J){$J$};
        \node at(1.5, 2.6)  [state] (K){$K$};
        \draw[-, red] (I) to[bend left=10]  node[] {} (J);
        \draw[-] (I) to[bend right=10]  node[below] {$v$} (J);
        \draw[-, red] (J) to[bend left=10]  node[] {} (K);
        % \draw[-] (J) to[bend right=10]  node[above] {v} (K);
        \draw[-, red] (K) to[bend left=10]  node[] {} (I);
        \draw[-] (K) to[bend right=10]  node[above] {$v$} (I);
    \end{tikzpicture}
    (c)
    \end{minipage}
    \caption{Three cases for AE-MonoEq$\Delta$. The red edges represent edges colors, and the black edges represent edge values.}
    \label{fig:monochromatic-equality-triangle-cases}
\end{figure}

Cases (b) and (c) are symmetric, so it suffices to only consider case (a) and case (b). The following theorem shows that we can solve AE-MonoEq$\Delta$ in $\tilde{O}(n^{(3+\omega)/2})$ time. Moreover, when $\omega > 2$, if AE-Mono$\Delta$ has a better algorithm, so does AE-MonoEq$\Delta$. However, when $\omega = 2$, we must have $\lambda = 0$ in the theorem, so AE-MonoEq$\Delta$ won't necessarily have a better algorithm given a better algorithm for AE-Mono$\Delta$.

\begin{theorem}\label{thm:monoeqT_to_monoT}
Suppose $\omega \ge 2 + \lambda$ for some $\lambda \ge 0$, and suppose AE-Mono$\Delta$ has an $O(n^{(3+\omega)/2 - \epsilon})$ time algorithm, then AE-MonoEq$\Delta$ has an \[\tilde{O}\left(n^{(\omega+3)/2 - \frac{\frac{\lambda}{2(\kappa - 1)}\epsilon}{\omega + 3 + \frac{\lambda}{2(\kappa - 1)} - 2\epsilon}}\right)\]
time algorithm, where $\kappa \ge 3$ is a constant depending on $\omega$ and $\lambda$. 
\end{theorem}
\begin{proof}
We first show the algorithm for AE-MonoEq$\Delta$ on tripartite graphs $I \cup J \cup K$ with values on edge sets $I \times K$ and $J \times K$ (Case (a) in Figure~\ref{fig:monochromatic-equality-triangle-cases}). We define a good triangle $(i, j, k)$ to be a triangle where all its three edges have the same color, and $v(i, k) = v(j, k)$. For each edge $e \in E \cap (I \times J)$, the algorithm needs to report whether it is in a good triangle or not.

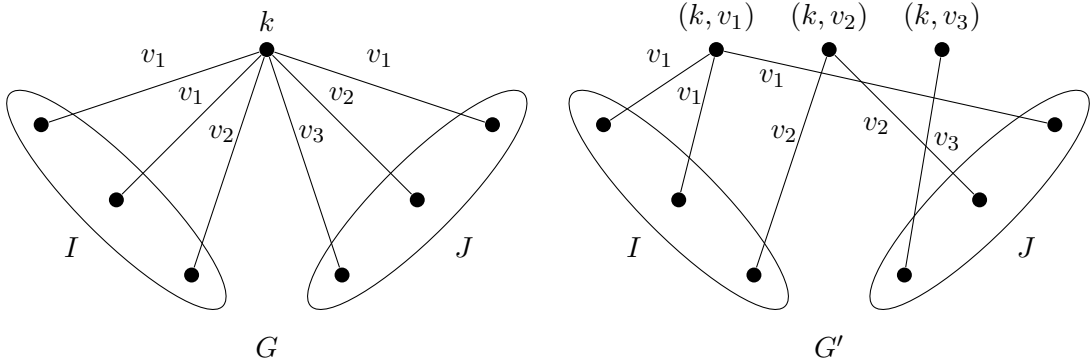
\begin{figure}[ht]
    \centering
    \begin{minipage}[t]{210pt}
    \centering
    \begin{tikzpicture}
    state/.style ={ circle ,
        draw,black , text=black , inner sep=0pt, text width=3mm, align=center}]
        \node at(0, 0)  [circle,fill,inner sep=2pt,label=above:$k$] (k){};
        \node[ellipse, draw,align=left, xshift=-0.5cm, rotate=-45, minimum width = 4cm,minimum height = 1cm,label=below:$I$] at (-1.5, -2) (non) {};
        \node[ellipse, draw,align=left, xshift=-0.5cm, rotate=45, minimum width = 4cm,minimum height = 1cm,label=below:$J$] at (2.5, -2) (non) {};
        \node at(-3, -1)  [circle,fill,inner sep=2pt] (i1){};
        \node at(-2, -2)  [circle,fill,inner sep=2pt] (i2){};
        \node at(-1, -3)  [circle,fill,inner sep=2pt] (i3){};
        \node at(3, -1)  [circle,fill,inner sep=2pt] (j1){};
        \node at(2, -2)  [circle,fill,inner sep=2pt] (j2){};
        \node at(1, -3)  [circle,fill,inner sep=2pt] (j3){};
        \draw[-,] (k) to[]  node[label=above:$v_1$] {} (i1);
        \draw[-,] (k) to[]  node[label=above:$v_1$] {} (i2);
        \draw[-,] (k) to[]  node[label={[xshift=-0.1cm, yshift=0.0cm]$v_2$}] {} (i3);
        \draw[-,] (k) to[]  node[label=above:$v_1$] {} (j1);
        \draw[-,] (k) to[]  node[label=above:$v_2$] {} (j2);
        \draw[-,] (k) to[]  node[label={[xshift=0.1cm, yshift=0.0cm]$v_3$}] {} (j3);
    \end{tikzpicture}
    $G$
    \end{minipage}
    \begin{minipage}[t]{210pt}
    \centering
    \begin{tikzpicture}
    state/.style ={ circle ,
        draw,black , text=black , inner sep=0pt, text width=3mm, align=center}]
        \node at(-1.5, 0)  [circle,fill,inner sep=2pt,label=above:{$(k,v_1)$}] (k1){};
        \node at(0, 0)  [circle,fill,inner sep=2pt,label=above:{$(k,v_2)$}] (k2){};
        \node at(1.5, 0)  [circle,fill,inner sep=2pt,label=above:{$(k,v_3)$}] (k3){};
        
        \node[ellipse, draw,align=left, xshift=-0.5cm, rotate=-45, minimum width = 4cm,minimum height = 1cm,label=below:$I$] at (-1.5, -2) (non) {};
        \node[ellipse, draw,align=left, xshift=-0.5cm, rotate=45, minimum width = 4cm,minimum height = 1cm,label=below:$J$] at (2.5, -2) (non) {};
        \node at(-3, -1)  [circle,fill,inner sep=2pt] (i1){};
        \node at(-2, -2)  [circle,fill,inner sep=2pt] (i2){};
        \node at(-1, -3)  [circle,fill,inner sep=2pt] (i3){};
        \node at(3, -1)  [circle,fill,inner sep=2pt] (j1){};
        \node at(2, -2)  [circle,fill,inner sep=2pt] (j2){};
        \node at(1, -3)  [circle,fill,inner sep=2pt] (j3){};
        \draw[-,] (k1) to[]  node[label=above:$v_1$] {} (i1);
        \draw[-,] (k1) to[]  node[label={[xshift=-0.1cm, yshift=0.0cm]$v_1$}] {} (i2);
        \draw[-,] (k2) to[]  node[label={[xshift=-0.1cm, yshift=0.0cm]$v_2$}] {} (i3);
        \draw[-,] (k1) to[]  node[pos=0.15, label={[xshift=0.0cm, yshift=-0.6cm]$v_1$}] {} (j1);
        \draw[-,] (k2) to[]  node[pos = 0.3, label=below:$v_2$] {} (j2);
        \draw[-,] (k3) to[]  node[pos = 0.4, label=right:$v_3$, xshift=-5pt] {} (j3);
    \end{tikzpicture}
    $G'$
    \end{minipage}
    
    \caption{Transforming graph $G$ to graph $G'$ by copying each vertex $k$ and adding edges with value $v$ to vertex $(k, v)$.}
    \label{fig:value-transforming}
\end{figure}
We will create a new graph $G'$ such that the edges of $G'$ only has colors, instead of having both colors and values. The vertex set of $G'$ will be $I \cup J \cup (K \times \mathcal{V})$, where $\mathcal{V}$ is the set of all values in $G$ (we will create the graph lazily, so we won't actually spend time creating those vertices that end up not having any neighbors). The edge set between $I$ and $J$ in $G'$ is the same as the edge set between $I$ and $J$ in $G$. For each edge $(i, k)$ in $E(G)$ with value $v$ and color $c$, we create an edge between $i$ and $(k, v)$ in $G'$ with the same color $c$. Similarly, for each edge $(j, k)$, in $E(G)$ with value $v$ and color $c$, we create an edge between $j$ and $(k, v)$ in $G'$ with the same color $c$. We can see that edge $(i, j)$ is in a good triangle in $G$ if and only if it is in a monochromatic triangle in $G'$. Thus, from now on, we can focus on the All-Edges Monochromatic Triangle problem on the unbalanced graph $G'$.

Let $D \ge 1$ and $T \ge n$ be some parameters to be fixed later. For each color $c$, let $G'_c$ be the subgraph of $G'$ consisting of vertices of $G'$ and all edges with color $c$. For each vertex $(k, v) \in K'$, if its degree in $G'_c$ is at most $D$, then we can enumerate all pairs of its neighbors and test if these three vertices form a triangle in $\tilde{O}(1)$ time. If a triangle $(i, j, (k, v))$ is found during the enumeration, we record that edge $(i, j)$ is in a good triangle. After the enumeration, we can remove the vertex $(k, v)$ together with all edges incident to it from the graph $G'_c$. 
Summing over all $c$ and all $(k, v)$, these enumerations take $O(n^2 D)$ time. We can now assume that, in each $G'_c$, the degree of every vertex $(k, v) \in K'$ is at least $D$. 

We consider two cases, based on the the remaining size of $K'$ in $G'_c$. 
As a high level description, if the size of $K'$ is at least $T$, we use rectangular matrix multiplication. When the size of $K'$ is less than $T$, we combine all instances with a small size of $K'$ across every $G'_c$ into a single AE-Mono$\Delta$ instance, and use the assumed $O(n^{(3+\omega)/2  -\epsilon})$ time algorithm for AE-Mono$\Delta$. We describe and analyze these two cases in more details in the following. 

\textbf{Rectangular matrix multiplication for more unbalanced colors}

For a subgraph $G_c'$, if $K'$ has size at least $T$, we use rectangular matrix multiplication to determine whether each edge $e \in E(G_c') \cap (I \times J)$ is in a triangle. To do so, we create an integer matrix $X$ of dimension $|I| \times |K'|$ and an integer matrix $Y$ of dimension $|K'| \times |J|$. Initially all entries in $X$ and $Y$ are zero. If there is an edge between $i \in I$ and $k' \in K'$, we set the $(i, k')$-th entry of $X$ to $1$. Similarly, if there is an edge between $j \in J$ and $k' \in K'$, we set the $(k', j)$-th entry of $Y$ to $1$. It is then not hard to see that edge $(i, j)$ is in a triangle if and only if $(XY)_{i, j} >0 $. Since $|I|, |J| \le n$, it would take $O(n^{\omega(1, \log_n |K'|, 1)})$ time to multiply $X$ and $Y$.

In order to analyze rectangular matrix multiplication, we need some bound on $\omega(1, t, 1)$ in the regime when $\omega \ge 2 + \lambda$. 
\begin{claim}
\label{cl:rect-MM-time}
There exists a constant $\kappa$ such that $\omega(1, t, 1) \le \omega + (t-1)(1-\frac{\lambda}{2(\kappa-1)})$ for every $1 \le t \le 3$. 
\end{claim}
\begin{proof}
If $\omega = 2$, then $\lambda$ must be $0$ and thus the claim is trivially true. Now assume $\omega > 2$. 

Using $\delta = \omega - 2 - \lambda / 2$ in Corollary~\ref{cor:rect-approach-inf}, there exists $\kappa \ge 3$ such that $\omega(1, \kappa, 1) \le \omega + \kappa - 1 -\lambda/2$. By convexity of $\omega(1, *, 1)$, for any $1 \le t \le \kappa$, $\omega(1, t, 1) \le \frac{t-1}{\kappa-1} (\omega + \kappa - 1 - \lambda / 2) + \frac{\kappa - t}{\kappa - 1} \omega = \omega + (t-1)(1-\frac{\lambda}{2(\kappa-1)})$. 

\end{proof}

Let $t = \log_n T$. If the size of $K'$ in $G_c'$ is $n^{q_c}$ for some $q_c\ge t$, then it takes $O(n^{\omega(1, q_c, 1)})$ time to compute $XY$. Clearly $1 \le t \le q_c \le 3$, so $n^{\omega(1, q_c, 1)}=O(n^{\omega+(q_c-1)(1-\frac{\lambda}{2(\kappa-1))}})$ by Claim~\ref{cl:rect-MM-time}. Since the degree of every vertex in $K'$ is at least $D$, and the total number of edges across all $G_c'$ is $O(n^2)$, so the sum of the sizes of $K'$ is at most $O(\frac{n^2}{D})$. In other words, $\sum_c n^{q_c} = O(\frac{n^2}{D})$. Therefore, the overall time complexity can be bounded as the following up to constant factors. 
\[
\begin{split}
    \sum_c n^{\omega+(q_c-1)(1-\frac{\lambda}{2(\kappa-1)})} &= n^\omega \cdot \sum_c n^{(q_c - 1)} / n^{\frac{(q_c - 1) \lambda}{ 2(\kappa - 1)}}\\
    &\le n^\omega \cdot \sum_c n^{(q_c - 1)} / n^{\frac{(t - 1) \lambda}{ 2(\kappa - 1)}}\\
    &= \frac{n^{\omega + 1 + \frac{\lambda}{2(\kappa - 1)}}}{D T^{\frac{\lambda}{2(\kappa - 1)}}}.
\end{split}
\]
Therefore, this case takes at most $O(\frac{n^{\omega + 1 + \frac{\lambda}{2(\kappa - 1)}}}{D T^{\frac{\lambda}{2(\kappa - 1)}}})$ time. 

\textbf{All-Edges Monochromatic Triangle for moderately unbalanced colors}

In this part, we consider colors $c$ such that the $K'$ part in $G_c'$ has at most $T$ vertices. Since $T\ge n$ we can assume the number of vertices in $G_c'$  is at most $2n + T \le 3T$. Also, if we sum over the number of edges of $G_c'$ of every $c$, the total number of edges is at most $n^2$. 

Note that the graph on each $G_c'$ is an All-Edges Sparse Triangle instance. 
We can actually combine all these instances to a single All-Edges Monochromatic Triangle instance of size $\tilde{O}(T)$ by a similar reduction of the proof of Theorem~\ref{thm:sparse_to_mono}.  

The reduction from All-Edges Sparse Triangle to All-Edges Monochromatic Triangle in 
Theorem~\ref{thm:sparse_to_mono} works as follows \cite{lincoln2020monochromatic}. Let $H$ be a multi-graph on $3T$ vertices, initially with no edges. 
For each All-Edges Sparse Triangle instance on vertex set $[3T]$, we take a random permutation of its vertices, and copy each edge in this instance to $H$, and color these edges using a color unique to  this instance. On expectation, for every $(u, v) \in [3T] \times [3T]$, the multiplicity of $(u, v)$ in $H$ is $\frac{n^2}{(3T)^2} = O(1)$. Since the permutations for graphs $G_c'$ are independent for different colors $c$, by Chernoff bound, the multiplicity of any $(u, v)$ is $O(\log n)$ with high probability. For every $(u, v)$ in $[3T] \times [3T]$, we arbitrarily label its corresponding $O(\log n)$ edges using numbers from $1$ to $O(\log n)$. Then we enumerate $O(\log^3 n)$ triples $(i, j, k)$ of labels. For each triple $(i, j, k)$, we create a graph on vertex set $V_1 \sqcup V_2 \sqcup V3$, where $V_1 = V_2 = V_3 = [3T]$. Between $V_1$ and $V_2$, we add edges in $H$ with label $i$; between $V_2$ and $V_3$, we add edges with label $j$; between $V_3$ and $V_1$, we add edges with label $k$. Then a triangle in one All-Edges Sparse Triangle instance corresponds to some monochromatic triangle in one of these $O(\log^3 n)$ instances. This finishes the reduction.

For each of the $O(\log^3 n)$ instances of All-Edges Monochromatic Triangle, we can use the assumed $O(n^{(3+\omega)/2 - \epsilon})$ time algorithm for All-Edges Monochromatic Triangle, so we get an $\tilde{O}(T^{(3+\omega) / 2 -\epsilon})$ time algorithm for this case. 

\textbf{Running Time}
The three cases have running times $O(n^2D)$, $O(\frac{n^{\omega + 1 + \frac{\lambda}{2(\kappa - 1)}}}{D T^{\frac{\lambda}{2(\kappa - 1)}}})$, and $\tilde{O}(T^{(3+\omega) / 2 -\epsilon})$ respectively. We can balance them by setting $T=n^{1+\frac{2\epsilon}{\omega+3+\frac{\lambda}{2(\kappa - 1)}-2\epsilon}}$ and $D = \frac{T^{(3+\omega)/2-\epsilon}}{n^2}$. The overall running time is $$\tilde{O}\left(n^{(\omega+3)/2 - \frac{\frac{\lambda}{2(\kappa - 1)}\epsilon}{\omega + 3 + \frac{\lambda}{2(\kappa - 1)} - 2\epsilon}}\right).$$

Now we consider the case where the values are on edge sets $I \times J$ and $I \times K$ (Case (b) in Figure~\ref{fig:monochromatic-equality-triangle-cases}). The case where the values are on edge sets $I \times J$ and $J \times K$ (Case (c) in Figure~\ref{fig:monochromatic-equality-triangle-cases}) is symmetric. The algorithm is largely the same with some small changes in details. Instead of creating graph $G'$ with vertex set $I \cup J \cup (K \times \mathcal{V})$, we create graph $G'$ with vertex set $(I' = I \times \mathcal{V}) \cup J \cup K$. Similarly, now each vertex in $I'$ represents both a vertex $i \in I$ and a value. The edges in $G'$ are added similarly to the algorithm for Case (a). 

Now it is sufficient to compute whether each edge in $I' \times J$ is in a monochromatic triangle in $G'$. Similarly, we partition the edge set by the color of the edges to get instances $G'_c$. For those vertices $i' \in I'$ that has degree at most $D$ in some $G'_c$, we could enumerate all pairs of neighbors of $i'$ and test if any pair forms a triangle in $O(\deg^2_{G'_c}(i'))$ time. Overall, this case takes $O(n^2D)$ time. 

For the remaining vertices, we similarly consider the size of $I'$. If $|I'| = n^{p_c} \ge T$, we use rectangular matrix multiplication. In this case, we will compute the product of an $n^{p_c}$ by $n$ matrix, and an $n$ by $n$ matrix, so it takes $O(n^{\omega(p_c, 1, 1)})$ time. Note that $\omega(p_c, 1, 1) = \omega(1, p_c, 1)$, so this case runs in the same time as the corresponding case in algorithm for Case (a). 

Finally, for those $G_c'$ whose $I'$ has size at most $T$, we can still combine them into $\tilde{O}(1)$ All-Edges Monochromatic Triangle instances. Thus, the running time  is still the same as the running time for the corresponding case in algorithm for Case (a). 

Since all three cases share the same running times, the overall running time remains the same.
% $$\tilde{O}\left(n^{(\omega+3)/2 - \frac{\frac{\lambda}{2(\kappa - 1)}\epsilon}{\omega + 3 + \frac{\lambda}{2(\kappa - 1)} - 2\epsilon}}\right).$$

\end{proof}

% \begin{definition}[$(\min, =)$-product]
% Given two $n \times n$ matrices $A$ and $B$, compute the matrix $C$ such that $C_{i, j} = \min\{B_{k, j} | A_{i, k} = B_{k, j}\}$. 
% \end{definition}

% \begin{definition}[$(\min, \le)$-product]
% Given two $n \times n$ matrices $A$ and $B$, compute the matrix $C$ such that $C_{i, j} = \min\{B_{k, j} | A_{i, k} \le B_{k, j}\}$. 
% \end{definition}

% \begin{definition}[$(\max, \le)$-product]
% Given two $n \times n$ matrices $A$ and $B$, compute the matrix $C$ such that $C_{i, j} = \max\{B_{k, j} | A_{i, k} \le B_{k, j}\}$. 
% \end{definition}

Next, we show reductions to AE-MonoEq$\Delta$ from many other intermediate problems. 

\begin{theorem}
\label{thm:min-equal-to-mono-equal-triangle}
If there is a $T(n)$ time algorithm for AE-MonoEq$\Delta$, then there is an $O(T(n) \log n)$ time algorithm for $(\min, =)$-product. 
\end{theorem}
\begin{proof}
We can add a column to matrix $A$ with an entry value that's larger than all other entries, and add the corresponding row to matrix $B$ with the same value. This value won't affect $C_{i, j}$ when $\{B_{k, j}| A_{i, k} = B_{k, j}\}_k$ is nonempty. If the computed $C_{i, j}$ equals this large value, we know $\{B_{k, j}| A_{i, k} = B_{k, j}\}_k$ is empty, and can then set $C_{i, j}$ back to $\infty$. Thus, we can assume $C_{i,j}$ is finite for every $i, j$. 

We can easily discretize the entries of $A$ and $B$, so that we can assume all entries are integers between $0$ and $2n^2-1$. 

We create a complete tripartite graph $G$ with vertex set $I, J, K$. For each edge $(i, k) \in I \times K$, we assign it a value $A_{i, k}$; for each edge $(j, k) \in J \times K$, we assign it a value $B_{k, j}$. 
We will use colors on the edges to perform parallel binary search to find the $(\min, =)$-product. 

Let $t= \lceil \log_2 (2n^2)\rceil$. 
Initially, we know that $C_{i, j} \in [0, 2^t)$ for every $i, j$. 
We will perform $t$ calls to the $T(n)$ time algorithm for AE-MonoEq$\Delta$, and each call narrows the possible range of $C_{i, j}$ by a half. For each $0 \le \ell \le t$, we will compute $\tilde{C}^\ell_{i, j}$ such that  $\tilde{C}^\ell_{i, j}$ is a multiple of $2^{\ell}$ and $\tilde{C}^\ell_{i, j} \le C_{i, j} < \tilde{C}^\ell_{i, j} + 2^{\ell}$. The initial condition, when $\ell = t$, is clearly satisfied by setting $\tilde{C}^t_{i, j} = 0$.

Assume we have  have computed $\tilde{C}^{\ell+1}$. 
We will recolor the edges in $G$ for computing $\tilde{C}^{\ell}$. 
We set the color of an edge $(i, k) \in I \times K$ as $\lfloor \frac{A_{i, k}}{2^{\ell}}\rfloor$, and set the color of an edge $(j, k) \in J \times K$ as $\lfloor \frac{B_{k, j}}{2^{\ell}} \rfloor$. For each edge $(i, j)$, we set its color to $\lfloor \frac{\tilde{C}^{\ell+1}_{i, j}}{2^{\ell}} \rfloor$. Now we use the $T(n)$ time algorithm for AE-MonoEq$\Delta$ on this graph $G$. If edge $(i, j)$ is in a monochromatic equality triangle, then we set $\tilde{C}^{\ell} = \tilde{C}^{\ell+1}$; otherwise, we set $\tilde{C}^\ell = \tilde{C}^{\ell+1} + 2^{\ell}$.

We show that the values $\tilde{C}^\ell$ satisfy the conditions, assuming $\tilde{C}^{\ell+1}$ satisfies the conditions. First, since $\tilde{C}^{\ell+1}_{i, j}$ is a multiple of $2^{\ell + 1}$, $\tilde{C}^\ell_{i, j}$ will be a multiple of $2^{\ell}$ in either case. Suppose an edge $(i, j)$ is in a monochromatic equality triangle $(i, k, j)$. Since the values on edges $(i, k)$ and $(k, j)$ are the same, we have $A_{i, k} = B_{k, j}$. Also, the colors on edges $(i, k), (k, j), (i, j)$ are the same, so $\lfloor \frac{A_{i, k}}{2^{\ell}}\rfloor = \lfloor \frac{B_{k, j}}{2^{\ell}}\rfloor = \lfloor \frac{\tilde{C}^{\ell+1}_{i, k}}{2^{\ell}}\rfloor$. Since $\tilde{C}^{\ell+1}$ is a multiple of $2^{\ell+1}$, we further get that $\tilde{C}^{\ell+1}_{i, j} \le A_{i, k} = B_{k, j} < \tilde{C}^{\ell+1}_{i, j} + 2^{\ell}$. We can similarly show that if there exists $k$ such that $\tilde{C}^{\ell+1}_{i,j} \le A_{i, k} = B_{k, j} < \tilde{C}^{\ell+1}_{i,j} + 2^{\ell}$, then $(i, j)$ is on a monochromatic equality triangle. 
Thus, $(i, j)$ is on a monochromatic equality triangle if and only if $\tilde{C}^{\ell+1}_{i,j} \le C_{i, j} < \tilde{C}^{\ell+1}_{i,j} + 2^{\ell}$. 
Therefore, if $(i, j)$ is on a monochromatic equality triangle, setting $\tilde{C}^\ell_{i, j} = \tilde{C}^{\ell + 1}_{i, j}$ satisfies $\tilde{C}^{\ell}_{i, j} \le C_{i, j} < \tilde{C}^{\ell}_{i,j} + 2^{t - \ell}$; otherwise, since $\tilde{C}^{\ell+1}_{i,j} \le C_{i, j} < \tilde{C}^{\ell+1}_{i,j} + 2^{\ell+1}$, we must have $\tilde{C}^{\ell+1}_{i,j} + 2^{\ell} \le C_{i, j} < \tilde{C}^{\ell+1}_{i,j} + 2^{\ell} +2^{\ell}$, so setting $\tilde{C}^\ell_{i,j} = \tilde{C}^{\ell+1}_{i,j} + 2^{\ell}$ satisfies the conditions. 

Now if we performed all $t$ rounds, we would get $\tilde{C}^t_{i, j}$ such that $\tilde{C}^{t}_{i, j} \le C_{i, j} < \tilde{C}^t_{i, j} + 2^{t-t}$, so $\tilde{C}^t_{i, j} = C_{i,j}$.

\end{proof}

\begin{theorem}
\label{thm:min-less-to-mono-equal-triangle}
If there is a $T(n)$ time algorithm for AE-MonoEq$\Delta$, then there is an $O(T(n) \log^2 n)$ time algorithm for $(\min, \le)$-product. 
\end{theorem}
\begin{proof}

First, we can add $1$ to every entry of $B$, so now the problem becomes a $(\min, <)$-product. Also, we can easily discretize the entries of $A$ and $B$, so that we can assume all entries are integers between $0$ and $2n^2-1$. 

For every $A_{i, k} < B_{k, j}$, there is some integer $\ell$ such that the bit corresponding to $2^{\ell-1}$ in $A_{i, k}$'s binary representation is $0$, the bit corresponding to $2^{\ell-1}$ in $B_{k, j}$'s binary representation is $1$, and $\lfloor \frac{A_{i, k}}{2^{\ell}} \rfloor = \lfloor \frac{B_{k, j}}{2^{\ell}} \rfloor$. Our algorithm enumerates this $\ell$, and handles different $\ell$ independently. 

Fix some $1 \le \ell \le \lceil \log_2 (2n^2)\rceil$. We aim to compute $C^\ell$, where $C^\ell_{i, j}$ is defined as $$C^\ell_{i, j} = \min_k \{B_{k, j} | A_{i, k} < B_{k, j} \wedge \text{ the highest differing bit of } A_{i, k} \text{ and } B_{k, j} \text{ is the bit corresponding to } 2^{\ell-1}\}.$$

We create new matrices $\tilde{A}^\ell$ and $\tilde{B}^\ell$. For some $i, k$, if the bit corresponding to $2^{\ell-1}$ in the binary representation of $A_{i, k}$ is $0$, we set $\tilde{A}^\ell_{i, k}$ to  $\lfloor \frac{A_{i, k}}{2^\ell} \rfloor$; otherwise, we set $\tilde{A}^\ell_{i, k}$ to $-1$. Similarly, For some $k, j$, if the bit corresponding to $2^{\ell-1}$ in the binary representation of $B_{k, j}$ is $1$, we set $\tilde{B}^\ell_{k, j}$ to  $\lfloor \frac{B_{k, j}}{2^\ell} \rfloor$; otherwise, we set $\tilde{B}^\ell_{k, j}$ to $-2$. Then we use Theorem~\ref{thm:min-equal-to-mono-equal-triangle} to compute the $(\min, =)$-product $\tilde{C}^\ell$ of $\tilde{A}^\ell$ and $\tilde{B}^\ell$ in $O(T(n) \log n)$ time. If $\tilde{C}^\ell_{i, j} < \infty$, then clearly $\tilde{C}^\ell_{i,j} \cdot 2^\ell \le C^\ell_{i,j} < \tilde{C}^\ell_{i,j} \cdot 2^\ell + 2^\ell$; if $\tilde{C}^\ell_{i, j} = \infty$, then $C^\ell_{i, j}$ is also $\infty$. 

For every $i, j$, it suffices to find $\min_k \{B_{k,j} | \tilde{A}^\ell_{i, k} = \tilde{B}^\ell_{k, j} = \tilde{C}^\ell_{i, j}\}$. We can use the parallel binary search idea again. Create a complete tripartite graph $G$ with vertex set $I \cup J \cup K$. For edge $(i, j) \in I \times J$, we use $\tilde{C}^\ell_{i, j}$ as its color; for edge $(i, k) \in I \times K$, we use $\tilde{A}^\ell_{i, k}$ as its color; for edge $(k, j) \in K \times J$, we use $\tilde{B}^\ell_{k, j}$ as its color. The values of the graph will be on edge set $(I \times J) \cup (J \times K)$. 

For every $r \le \ell$, we will compute an estimate $\tilde{C}^{\ell, r}_{i, j}$ such that $\tilde{C}^{\ell, r}_{i, j}$ is a multiple of $2^r$ and $\tilde{C}^{\ell, r}_{i, j} \le C^\ell_{i, j} < \tilde{C}^{\ell, r}_{i, j} + 2^r$. When $r=\ell$, we can clearly set $\tilde{C}^{\ell, r}_{i, j} = \tilde{C}^\ell_{i,j} \cdot 2^\ell$. Now suppose we have computed $\tilde{C}^{\ell, r + 1}$ and want to compute $\tilde{C}^{\ell, r}$. On the graph $G$, we set the value of edge $(i, j) \in I \times J$ to be $\lfloor \frac{\tilde{C}^{\ell, r + 1}}{2^r}\rfloor$, and set the value of edge $(j, k) \in J \times K$ to be $\lfloor \frac{B_{k, j}}{2^r} \rfloor$. Then we use the $T(n)$ time AE-MonoEq$\Delta$ algorithm on graph $G$. For every $(i, j)$, if it is on a monochromatic equality triangle, we set $\tilde{C}^{\ell, r}_{i, j} = \tilde{C}^{\ell, r+1}_{i, j}$; otherwise, we set $\tilde{C}^{\ell, r}_{i, j} = \tilde{C}^{\ell, r+1}_{i,j} + 2^r$. 

Clearly, $\tilde{C}^{\ell, r}_{i, j}$ is a multiple of $2^r$, since $\tilde{C}^{\ell, r+1}_{i,j}$ is a multiple of $2^{r+1}$. If $(i, j)$ is on a monochromatic equality triangle, then there exists $k$ such that $\lfloor \frac{B_{k, j}}{2^r} \rfloor = \lfloor \frac{\tilde{C}^{\ell, r + 1}}{2^r}\rfloor$ and $\tilde{A}^\ell_{i, k} = \tilde{B}^\ell_{k, j} = \tilde{C}^\ell_{i, j}$. Also because $\tilde{C}^{\ell, r + 1}$ is a multiple of $2^{r+1}$, we must have $\tilde{C}^{\ell, r + 1} \le B_{k, j} < \tilde{C}^{\ell, r + 1} + 2^r$. Since $C^\ell_{i, j} \ge \tilde{C}^{\ell, r + 1}$, we must have $\tilde{C}^{\ell, r + 1} \le C^\ell_{i, j} < \tilde{C}^{\ell, r + 1} + 2^r$. Thus, it is valid to set $\tilde{C}_{i, j}^{\ell, r} = \tilde{C}_{i, j}^{\ell, r+1}$ in this case.

If  $(i, j)$ is not on a monochromatic equality triangle, then we can similarly show that the best $B_{k, j}$ where $\tilde{A}^\ell_{i, k} = \tilde{B}^\ell_{k, j} = \tilde{C}^\ell_{i, j}$ must be  at least $\tilde{C}^{\ell, r + 1} + 2^r$. Also, by the guarantee of $\tilde{C}^{\ell, r + 1}_{i, j}$, the best $B_{k, j}$ must be smaller than $\tilde{C}^{\ell, r + 1}_{i, j}+2^{r+1}$. 
Thus, it is valid to set $\tilde{C}_{i, j}^{\ell, r} = \tilde{C}_{i, j}^{\ell, r+1} + 2^r$ in this case since
$\tilde{C}^{\ell, r + 1}_{i, j} + 2^r \le C^\ell_{i, j} < \tilde{C}^{\ell, r + 1}_{i, j} + 2^r + 2^r$. 

After we compute $\tilde{C}^{\ell, r}$ for all $0 \le r \le \ell$, we can simply set $C^\ell = \tilde{C}^{\ell, 0}$, since the guarantee of $\tilde{C}^{\ell, 0}$ is $\tilde{C}^{\ell, 0}_{i, j} \le \tilde{C}^\ell_{i, j} < \tilde{C}^{\ell, 0}_{i, j}+2^0$.

After we compute $C^\ell$ for every $\ell$, we can compute $C_{i, j} = \min_\ell C^\ell_{i, j}$. The overall time complexity is $O(T(n) \log^2 n)$ since the number of bit $\ell$ is $O(\log n)$, and computing $C^\ell$ for each $\ell$ takes $O(T(n) \log n)$ time. 
\end{proof}

Using a similar proof, we can get a reduction to AE-MonoEq$\Delta$ from $(\max, \le)$ product. Even though $(\max, \le)$-product looks similar to $(\min, \le)$-product, and their best algorithms both run in $\tilde{O}(n^{(3+\omega)})$ time \cite{duanpettiebott}, we  don't know if they are equivalent. 
\begin{proposition}
\label{prop:max-less-to-mono-equal-triangle}
If there is a $T(n)$ time algorithm for AE-MonoEq$\Delta$, then there is an $O(T(n) \log^2 n)$ time algorithm for $(\max, \le)$-product. 
\end{proposition}

Now we consider the Monochromatic Equality Product problem, which can be viewed as Case (a) of AE-MonoEq$\Delta$. 
% \begin{definition}[Monochromatic Equality Product]
% Given a graph $G$ on vertex sets $I \cup J \cup K$. Each edge $e$ in the graph has a color $c(e)$. All edges $e$ in $I \times K$ and $J \times K$ has a value $v(e) $. For every $(i, j)$, decide if there exists $k$ such that $v[i, j] = v[j, k]$ and $c[i, k] = c[j, k] = c[i, j]$.
% \end{definition}

Note that the proof of Theorem~\ref{thm:min-equal-to-mono-equal-triangle} only uses Case (a) of AE-MonoEq$\Delta$, so the same proof actually shows a reduction from $(\min, =)$-product to Monochromatic Equality Product. In fact, we will show that Monochromatic Equality Product is equivalent to Monochromatic $(\min, =)$-product which is a stronger version of $(\min, =)$-product.

% \begin{definition}[Monochromatic $(\min, =)$-product]
% Given a graph $G$ on vertex sets $I \cup J \cup K$. Each edge $e$ in the graph has a color $c(e)$. All edges $e$ in $I \times K$ and $J \times K$ has a value $v(e) $. For every $(i, j)$, compute 
% $$\min \{v[i, k] : v[i, k] = v[j, k] \wedge c[i, k] = c[j, k] = c[i, j] \forall k\} \cup \{\infty\}.$$
% \end{definition}

% Monochromatic $(\min, =)$-product can be viewed as a combination of multiple sparse Min Equality Product instances. 

The best algorithm for $(\min, =)$-product runs in $\tilde{O}(n^{(3+\omega)/2}) = \tilde{O}(n^{2.687})$ time, while the best algorithm for Equality Product runs in $\tilde{O}(n^{2.6598})$ time \cite{gold2017dominance}, where the improvement is brought by rectangular matrix multiplication. Therefore, we don't know if Min Equality product is equivalent to Equality Product. The following theorem suggests that if we add the Monochromatic constraint to both problems, they become equivalent up to poly-logarithmic factors.

\begin{theorem}
\label{thm:monomineq}
If there is a $T(n)$ time algorithm for Monochromatic $(\min, =)$-product then there is an $O(T(n))$ time algorithm for Monochromatic Equality Product. Also, if there is a $T(n)$ time algorithm for Monochromatic Equality Product, then there is an $O(T(n) \log n)$ time algorithm for Monochromatic $(\min, =)$-product. 

\end{theorem}
\begin{proof}
The first direction is trivially true, since Monochromatic $(\min, =)$-product computes more information than Monochromatic Equality Product. The second direction is more interesting.

Let $\mathbb{A}$ be an algorithm for Monochromatic Equality Product. 
Suppose we have an instance of Monochromatic $(\min, =)$-product of $n \times n$ matrices, with vertex sets $I, J, K$, edge colors $c(\cdot, \cdot)$, and edge values $A_{i, k}$ for $(i, k) \in I \times K$ and edge values $B_{k, j}$ for $(k, j) \in K \times J$. 
Clearly, we can discretize all the values so that they are integers between $[0, 2n^2)$. Let $C_{i, j}$ be the minimum value of $A_{i, k}$ such that $A_{i, k} = B_{k, j}$ and $c(i, k) = c(k, j) = c(i, j)$. Using $\mathbb{A}$, we can easily decide whether $C_{i, j}$ is $\infty$ for all pairs of $(i, j)$. Thus, we can focus on determining values for the finite entries of $C$ in the following.

We use the parallel binary search idea from before. Let $t = \lceil \log(2n^2)\rceil$. For each integer $0 \le \ell \le t$, we aim to compute $\tilde{C}_{i, j}^\ell$ so that $\tilde{C}_{i, j}^\ell$ is a multiple of $2^{\ell}$, and $\tilde{C}_{i, j}^\ell \le C_{i, j} < \tilde{C}_{i, j}^\ell + 2^\ell$. $\tilde{C}_{i, j}^t$ is easy to compute, since we can just set $\tilde{C}_{i, j}^t = 0$. 

Suppose for some $0 \le \ell < t$, we have computed $\tilde{C}^{\ell+1}$, we will compute $\tilde{C}^\ell$. In order to perform the binary search, we only need to know for each pair of $i, j$, whether there exists $k$ such that $c(i, k) = c(k, j) = c(i, j)$, $A_{i, k} = B_{k, j}$ and $\lfloor A_{i, k} / 2^\ell \rfloor = \lfloor B_{k, j} / 2^\ell \rfloor = \lfloor \tilde{C}^{\ell+1}_{i, k} / 2^\ell \rfloor$. If there exists one, then we can set $\tilde{C}_{i, j}^\ell = \tilde{C}_{i, j}^{\ell + 1}$; otherwise, we set $\tilde{C}_{i, j}^\ell = \tilde{C}_{i, j}^{\ell + 1} + 2^{\ell}$.

To determine the existence of such $k$, we use the algorithm for Monochromatic Equality Product.
We can create a Monochromatic Equality Product instance on the same vertex set. For each edge $(i, j) \in I \times J$, we set its color to $(c(i, j), \lfloor \tilde{C}^{\ell+1}_{i, k} / 2^\ell \rfloor)$; for $(i, k) \in I \times K$, we set its color to $(c(i, k), \lfloor A_{i, k} / 2^\ell \rfloor)$; for $(j, k) \in J\times K$, we set its color to $(c(j, k), \lfloor B_{k, j} / 2^\ell \rfloor)$. The values of the Monochromatic Equality Product is the same as the values of the original instance. Thus, if we use the $T(n)$ time algorithm $\mathbb{A}$ on this Monochromatic Equality Product instance, we would be able to compute $\tilde{C}^\ell$, and thus continue the binary search. 

After we compute $\tilde{C}^0$, we can easily set $C_{i, j} = \tilde{C}^0_{i, j}$ if $C_{i, j} < \infty$. Therefore, the algorithm runs in $O(T(n) \log n)$ time. 

\end{proof}

We further consider Monochromatic $(\min, \le)$-product, which turns out also has an $\tilde{O}(n^{(3+\omega)/2})$ time algorithm.

\begin{proposition}
\label{prop:monominle}
If there is a $T(n)$ time algorithm for AE-MonoEq$\Delta$, then there is an $O(T(n) \log^2 n)$ time algorithm for Monochromatic $(\min, \le)$ Product. 
\end{proposition}
The proof for Proposition~\ref{prop:monominle} is essentially a combination of the proof of Theorem~\ref{thm:min-less-to-mono-equal-triangle} and the ideas used in the proof of Theorem~\ref{thm:monomineq}, so we won't describe it in full detail for conciseness. For a high level description, we first reduce this problem to Monochromatic $(\min, <)$ Product with entries in $\{0, \ldots, 2n^2-1\}$. Then we enumerate the first differing bit in the binary representation of $A_{i, k}$ and $B_{k, j}$, and use Theorem~\ref{thm:monomineq} to find the smallest common binary prefix of $A_{i, k}$ and $B_{k, j}$. After we have this common prefix, we use it, together with the original color of the graph, as the color for a new graph. Also, we use edge values on $I \times J$ and $J \times K$ for performing parallel binary search, similar to what described in the proof of Theorem~\ref{thm:min-less-to-mono-equal-triangle}.

\section{Conclusion}
\begin{figure}[H]
    \centering
    \begin{tikzpicture}

        \node at(0, 0)  [] (monoT){AE-Mono$\Delta$};
     
        \node at(0, -1.5)  [] (monoeqT){AE-MonoEq$\Delta$};
        \node at(-3, -3)  [] (monoeqP){MonoEq};
        \node at(-6, -3)  [] (monomineqP){MonoMinEq};
        \node at(0, -3)  [] (minle){$(\min, \le)$};
        \node at(3, -3)  [] (maxle){$(\max, \le)$};
        
        \node at(-3, -4.5)  [] (mineq){$(\min, =)$};
        \node at(0, -4.5)  [] (maxmin){$(\max, \min)$};

        \draw[<-,dashed, bend left] (monoT.290) to[]  node[] {} (monoeqT.70);
        \draw[-Implies, double distance=2pt, bend right] (monoT.250) to[]  node[] {} (monoeqT.110);
        
        \draw[-Implies,double distance=2pt]  (monoeqP.north) to[]  node[] {} (monoeqT.230);
        \draw[Implies-Implies,double distance=2pt]  (monoeqP.west) to[]  node[] {} (monomineqP.east);
        \draw[-Implies,double distance=2pt]  (minle.north) to[]  node[] {} (monoeqT.south);
        \draw[-Implies,double distance=2pt] (maxle.north) to[]  node[] {} (monoeqT.310);
        
        \draw[-Implies,double distance=2pt] (mineq.north) to[]  node[] {} (monoeqP.south);
        \draw[-Implies,double distance=2pt] (maxmin.north) to[]  node[right] {\cite{VassilevskaWY07}} (minle.south);
        
        \draw[opacity=0.4, dashed, rounded corners=3] (current bounding box.north east) -- (current bounding box.north west) -- (current bounding box.south west) -- (current bounding box.south east) -- cycle;

        \node at(current bounding box.west)  [anchor=east] (){\textcolor{blue}{$\tilde{O}(n^{(3+\omega)/2})$}};

        \pgfmathsetmacro{\dx}{-7.35}
        \pgfmathsetmacro{\dy}{1.5}
        
        \node at(2 + \dx, 2 + \dy)  [anchor=center] (3sum){3SUM};
        \node at(3sum.west)  [anchor=east] (){\textcolor{blue}{$O(n^2)$}};
        
        \node at(0 + \dx, 0 + \dy)  [anchor=center] (apsp){APSP};
         \node at(apsp.west)  [anchor=east] (){\textcolor{blue}{$O(n^3)$}};
         
         \node at(0 + \dx, 0 + \dy)  [anchor=center] (apspholder){\begin{tabular}{c}  \\ \ \ \ \ \ \ \ \end{tabular}};
         
        \node at(2 + \dx, 0 + \dy)  [anchor=center] (exact){\begin{tabular}{c} Exact$\Delta$ \\ Zero$\Delta$ \end{tabular}};
        
        \node at(5.5 + \dx, 0 + \dy)  [anchor=west] (sparses){$[m^{1/3}] \times $ AE-Sparse$\Delta$};
        
        \node at(10 + \dx, 0 + \dy)  [anchor=west] (sparse){AE-Sparse$\Delta$};
        
        % \node at(7.35 + \dx, -1.5 + \dy)  [] (monoT){AE-Mono$\Delta$};

        \node at(5.5 + \dx, 1 + \dy)  [anchor=west] (inter){SetIntersection};
        
        \node at(5.5 + \dx, 2 + \dy)  [anchor=west] (disj){SetDisjointness};
        \node at(disj.east)  [anchor=west] (){\textcolor{blue}{$O(n^2)$}};

        %  \node at(0, -4)  [anchor=center] (exact){Exact$\Delta$};
         
        %  \node at(4.5, 0)  [anchor=west] (AElisting){AE $\Delta$ Listing};
        %  \node at(4.5, 2)  [anchor=west] (listing){$\Delta$Listing};
         
        %  \node at(7.25, 0.5)  [anchor=west] (AEdetection){AE-$\Delta$Detection};
         
        %  \node at(10.5, 0)  [anchor=west] (mono){AE-Mono$\Delta$};
        %  \node at(10.5, 1)  [anchor=west] (disjoint){SetDisjointness};
        %   \node at(10.5, 2)  [anchor=west] (intersect){SetIntersection};

        \draw[->,line width=1pt] (3sum.south) to[]  node[right] {\cite{VWfindingcountingj}} (exact.north);
        \draw[->,line width=1pt] (apsp.east) to[]  node[above] {\cite{focsyj}} (exact.west);
        \draw[->,line width=1pt] (exact.east) to[]  node[] {} (disj.west);
        \draw[->,line width=1pt] (exact.east) to[]  node[] {} (inter.west);
        \draw[->,line width=1pt] (exact.east) to[]  node[] {} node[pos=0.8,above]{$m^{5/3}$}(sparses.west);
        
        \draw[->,line width=1pt] (sparses.east) to[]  node[below] {\textit{trivial}} node[pos=0.8,above]{$m^{4/3}$}(sparse.west);
        
        \draw[->,line width=1pt] (sparses.south) to[]  node[left] {\cite{lincoln2020monochromatic}} node[pos=0.8,right]{$n^{5/2}$}(monoT.north);
        
        \draw[opacity=0.4, dashed, rounded corners=3] (3sum.north east) -- (3sum.north west) -- (3sum.south west) -- (3sum.south east) -- cycle;

        \draw[opacity=0.4, dashed, rounded corners=3] (exact.north east) -- (apspholder.north west) -- (apspholder.south west) -- (exact.south east) -- cycle;
        
         \draw[opacity=0.4, dashed, rounded corners=3] (disj.north east) -- (disj.north west) -- (inter.south west) -- (inter.south east) -- cycle;

    \end{tikzpicture}
    \caption{Main reductions in this paper. Single arrows represent normal fine-grained reductions. Double arrows represent that the running times before and after the reduction are the same up to poly-logarithmic factors. Dashed arrows represent reductions that hold only when $\omega > 2$. }
    \label{fig:conclusion}
\end{figure}
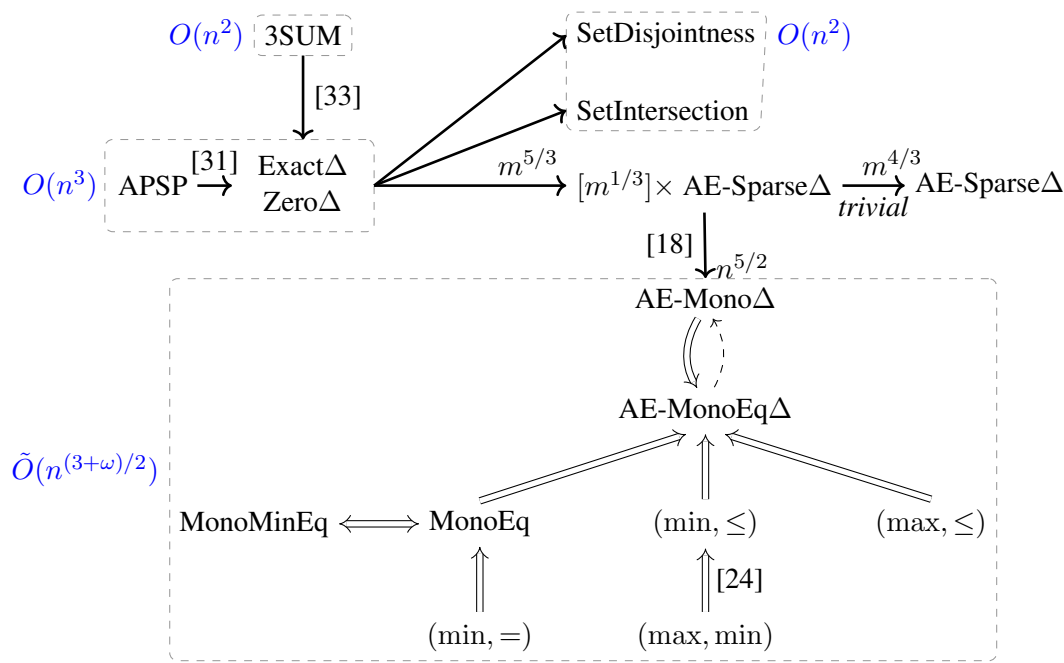

\bibliographystyle{plain}
\bibliography{ref}

\begin{thebibliography}{10}

\bibitem{abboud2018more}
Amir Abboud, Holger Dell, Karl Bringmann, and Jesper Nederlof.
\newblock More consequences of falsifying seth and the orthogonal vectors
  conjecture.
\newblock In {\em 50th Annual ACM Symposium on Theory of Computing, STOC 2018},
  pages 445--456. Association for Computing Machinery, Inc, 2018.

\bibitem{abboud2014popular}
Amir Abboud and Virginia~Vassilevska Williams.
\newblock Popular conjectures imply strong lower bounds for dynamic problems.
\newblock In {\em 2014 IEEE 55th Annual Symposium on Foundations of Computer
  Science}, pages 434--443. IEEE, 2014.

\bibitem{AlonGMN92}
Noga Alon, Zvi Galil, Oded Margalit, and Moni Naor.
\newblock Witnesses for boolean matrix multiplication and for shortest paths.
\newblock In {\em 33rd Annual Symposium on Foundations of Computer Science,
  Pittsburgh, Pennsylvania, USA, 24-27 October 1992}, pages 417--426. {IEEE}
  Computer Society, 1992.

\bibitem{AlonYZ97}
Noga Alon, Raphael Yuster, and Uri Zwick.
\newblock Finding and counting given length cycles.
\newblock {\em Algorithmica}, 17(3):209--223, 1997.

\bibitem{bjorklund2014listing}
Andreas Bj{\"o}rklund, Rasmus Pagh, Virginia~Vassilevska Williams, and Uri
  Zwick.
\newblock Listing triangles.
\newblock In {\em International Colloquium on Automata, Languages, and
  Programming}, pages 223--234. Springer, 2014.

\bibitem{CyganMWW19}
Marek Cygan, Marcin Mucha, Karol Wegrzycki, and Michal Wlodarczyk.
\newblock On problems equivalent to (min, +)-convolution.
\newblock {\em {ACM} Trans. Algorithms}, 15(1):14:1--14:25, 2019.

\bibitem{CzumajKL07}
Artur Czumaj, Miroslaw Kowaluk, and Andrzej Lingas.
\newblock Faster algorithms for finding lowest common ancestors in directed
  acyclic graphs.
\newblock {\em Theor. Comput. Sci.}, 380(1-2):37--46, 2007.

\bibitem{DuanJW19}
Ran Duan, Ce~Jin, and Hongxun Wu.
\newblock Faster algorithms for all pairs non-decreasing paths problem.
\newblock In {\em 46th International Colloquium on Automata, Languages, and
  Programming, {ICALP} 2019, July 9-12, 2019, Patras, Greece}, volume 132 of
  {\em LIPIcs}, pages 48:1--48:13. Schloss Dagstuhl - Leibniz-Zentrum f{\"{u}}r
  Informatik, 2019.

\bibitem{duanpettiebott}
Ran Duan and Seth Pettie.
\newblock Fast algorithms for (max, min)-matrix multiplication and bottleneck
  shortest paths.
\newblock In {\em Proceedings of the twentieth annual ACM-SIAM symposium on
  Discrete algorithms}, pages 384--391. SIAM, 2009.

\bibitem{duraj2020equivalences}
Lech Duraj, Krzysztof Kleiner, Adam Polak, and Virginia {Vassilevska Williams}.
\newblock Equivalences between triangle and range query problems.
\newblock In {\em Proceedings of the Fourteenth Annual ACM-SIAM Symposium on
  Discrete Algorithms}, pages 30--47. SIAM, 2020.

\bibitem{legallmult}
Fran{\c{c}}ois~Le Gall.
\newblock Powers of tensors and fast matrix multiplication.
\newblock In Katsusuke Nabeshima, Kosaku Nagasaka, Franz Winkler, and
  {\'{A}}gnes Sz{\'{a}}nt{\'{o}}, editors, {\em International Symposium on
  Symbolic and Algebraic Computation, {ISSAC} '14, Kobe, Japan, July 23-25,
  2014}, pages 296--303. {ACM}, 2014.

\bibitem{gold2017dominance}
Omer Gold and Micha Sharir.
\newblock Dominance product and high-dimensional closest pair under l\_infty.
\newblock In {\em 28th International Symposium on Algorithms and Computation
  (ISAAC 2017)}. Schloss Dagstuhl-Leibniz-Zentrum fuer Informatik, 2017.

\bibitem{huang1998fast}
Xiaohan Huang and Victor~Y Pan.
\newblock Fast rectangular matrix multiplication and applications.
\newblock {\em Journal of complexity}, 14(2):257--299, 1998.

\bibitem{kopelowitz2016higher}
Tsvi Kopelowitz, Seth Pettie, and Ely Porat.
\newblock Higher lower bounds from the 3sum conjecture.
\newblock In {\em Proceedings of the twenty-seventh annual ACM-SIAM symposium
  on Discrete algorithms}, pages 1272--1287. SIAM, 2016.

\bibitem{labib2019hamming}
Karim Labib, Przemys{\l}aw Uzna{\'n}ski, and Daniel Wolleb-Graf.
\newblock Hamming distance completeness.
\newblock {\em Leibniz International Proceedings in Informatics, LIPIcs}, 128,
  2019.

\bibitem{le2012faster}
Fran{\c{c}}ois Le~Gall.
\newblock Faster algorithms for rectangular matrix multiplication.
\newblock In {\em 2012 IEEE 53rd annual symposium on foundations of computer
  science}, pages 514--523. IEEE, 2012.

\bibitem{legallurr}
Fran{\c{c}}ois {Le Gall} and Florent Urrutia.
\newblock Improved rectangular matrix multiplication using powers of the
  coppersmith-winograd tensor.
\newblock In {\em Proceedings of the Twenty-Ninth Annual {ACM-SIAM} Symposium
  on Discrete Algorithms, {SODA} 2018, New Orleans, LA, USA, January 7-10,
  2018}, pages 1029--1046, 2018.

\bibitem{lincoln2020monochromatic}
Andrea Lincoln, Adam Polak, and Virginia {Vassilevska Williams}.
\newblock Monochromatic triangles, intermediate matrix products, and
  convolutions.
\newblock In {\em 11th Innovations in Theoretical Computer Science Conference
  (ITCS 2020)}. Schloss Dagstuhl-Leibniz-Zentrum f{\"u}r Informatik, 2020.

\bibitem{lotti1983asymptotic}
Grazia Lotti and Francesco Romani.
\newblock On the asymptotic complexity of rectangular matrix multiplication.
\newblock {\em Theoretical Computer Science}, 23(2):171--185, 1983.

\bibitem{MatIPL}
Ji{\v{r}}{\'{i}} Matou{\v{s}}ek.
\newblock Computing dominances in {E}{\^{}}n.
\newblock {\em Inf. Process. Lett.}, 38(5):277--278, 1991.

\bibitem{patrascu2010towards}
Mihai P{\u{a}}tra{\c{s}}cu.
\newblock Towards polynomial lower bounds for dynamic problems.
\newblock In {\em Proceedings of the forty-second ACM symposium on Theory of
  computing}, pages 603--610, 2010.

\bibitem{VassilevskaW09}
Virginia Vassilevska and Ryan Williams.
\newblock Finding, minimizing, and counting weighted subgraphs.
\newblock In Michael Mitzenmacher, editor, {\em Proceedings of the 41st Annual
  {ACM} Symposium on Theory of Computing, {STOC} 2009, Bethesda, MD, USA, May
  31 - June 2, 2009}, pages 455--464. {ACM}, 2009.

\bibitem{VassilevskaWY06}
Virginia Vassilevska, Ryan Williams, and Raphael Yuster.
\newblock Finding the smallest \emph{H}-subgraph in real weighted graphs and
  related problems.
\newblock In Michele Bugliesi, Bart Preneel, Vladimiro Sassone, and Ingo
  Wegener, editors, {\em Automata, Languages and Programming, 33rd
  International Colloquium, {ICALP} 2006, Venice, Italy, July 10-14, 2006,
  Proceedings, Part {I}}, volume 4051 of {\em Lecture Notes in Computer
  Science}, pages 262--273, 2006.

\bibitem{VassilevskaWY07}
Virginia Vassilevska, Ryan Williams, and Raphael Yuster.
\newblock All-pairs bottleneck paths for general graphs in truly sub-cubic
  time.
\newblock In David~S. Johnson and Uriel Feige, editors, {\em Proceedings of the
  39th Annual {ACM} Symposium on Theory of Computing, San Diego, California,
  USA, June 11-13, 2007}, pages 585--589. {ACM}, 2007.

\bibitem{VassilevskaWY10}
Virginia Vassilevska, Ryan Williams, and Raphael Yuster.
\newblock Finding heaviest \emph{H}-subgraphs in real weighted graphs, with
  applications.
\newblock {\em {ACM} Trans. Algorithms}, 6(3):44:1--44:23, 2010.

\bibitem{vstoc12}
Virginia {Vassilevska Williams}.
\newblock Multiplying matrices faster than {C}oppersmith-{W}inograd.
\newblock In Howard~J. Karloff and Toniann Pitassi, editors, {\em Proceedings
  of the 44th Symposium on Theory of Computing Conference, {STOC} 2012, New
  York, NY, USA, May 19 - 22, 2012}, pages 887--898. {ACM}, 2012.

\bibitem{vnotes2}
Virginia {Vassilevska Williams}.
\newblock Lecture nodes for lecture 8 of {CS367}, {O}ctober 15, 2015, 2015.

\bibitem{vnotes}
Virginia {Vassilevska Williams}.
\newblock Problem 2 on problem set 2 of {CS367}, {O}ctober 15, 2015, 2015.

\bibitem{Williams18}
R.~Ryan Williams.
\newblock Faster all-pairs shortest paths via circuit complexity.
\newblock {\em {SIAM} J. Comput.}, 47(5):1965--1985, 2018.

\bibitem{Williams14a}
Ryan Williams.
\newblock Faster all-pairs shortest paths via circuit complexity.
\newblock In David~B. Shmoys, editor, {\em Symposium on Theory of Computing,
  {STOC} 2014, New York, NY, USA, May 31 - June 03, 2014}, pages 664--673.
  {ACM}, 2014.

\bibitem{focsyj}
Virginia~Vassilevska Williams and R.~Ryan Williams.
\newblock Subcubic equivalences between path, matrix, and triangle problems.
\newblock {\em J. {ACM}}, 65(5):27:1--27:38, 2018.

\bibitem{focsy}
Virginia~Vassilevska Williams and Ryan Williams.
\newblock Subcubic equivalences between path, matrix and triangle problems.
\newblock In {\em 51th Annual {IEEE} Symposium on Foundations of Computer
  Science, {FOCS} 2010, October 23-26, 2010, Las Vegas, Nevada, {USA}}, pages
  645--654. {IEEE} Computer Society, 2010.

\bibitem{VWfindingcountingj}
Virginia~Vassilevska Williams and Ryan Williams.
\newblock Finding, minimizing, and counting weighted subgraphs.
\newblock {\em {SIAM} J. Comput.}, 42(3):831--854, 2013.

\bibitem{YusterDom}
Raphael Yuster.
\newblock Efficient algorithms on sets of permutations, dominance, and
  real-weighted {APSP}.
\newblock In Claire Mathieu, editor, {\em Proceedings of the Twentieth Annual
  {ACM-SIAM} Symposium on Discrete Algorithms, {SODA} 2009, New York, NY, USA,
  January 4-6, 2009}, pages 950--957. {SIAM}, 2009.

\bibitem{zwickbridge}
Uri Zwick.
\newblock All pairs shortest paths using bridging sets and rectangular matrix
  multiplication.
\newblock {\em J. {ACM}}, 49(3):289--317, 2002.

\end{thebibliography}

% \appendix
% \section{Non-Reducibility from SETH}
% \input{non_reducibility}

\end{document}